%% file: p.tex
\documentclass{vldb}

\newcommand{\paperTitle}{A Symbolic Approach to Proving Query Equivalence\\
Under Bag Semantics}
\newcommand{\paperKeywords}{Query Rewriting, Query Optimization, Satisfiability Modulo Theories}
\newcommand{\paperAuthors}{}

\usepackage[hyphens]{url}

\usepackage[usenames,dvipsnames]{xcolor}
\definecolor{linkcolor}{HTML}{647382}
\definecolor{citecolor}{HTML}{647382} %
\definecolor{urlcolor}{rgb}{0.4,0.2,0.2}
\definecolor{sqlcolor}{HTML}{965d67}
\definecolor{smtcolor}{HTML}{5d968c}
\definecolor{webblue}{rgb}{0,0,.7}
\definecolor{webgreen}{rgb}{0,.5,0}
\definecolor{webbrown}{rgb}{.6,0,0}
\definecolor{Red}{rgb}{1,0,0}

\usepackage[breaklinks]{hyperref}
\hypersetup{%
    pdfauthor = {\paperAuthors},
    pdftitle = {\paperTitle},
    pdfkeywords = {\paperKeywords},
    bookmarksopen = {true},
    colorlinks=true,
    citecolor={urlcolor},
    linkcolor={linkcolor},
 	urlcolor={citecolor},
    pdfborder={ 0 0 0 }
}

\usepackage{amsmath,amsopn,amssymb}
\usepackage{listings}
\usepackage{subfig}
\usepackage{endnotes,microtype,xspace,graphicx,fancyvrb,multirow}
\usepackage{supertabular,booktabs}
\usepackage{array,underscore,relsize}
\usepackage[T1]{fontenc}
\usepackage{times}
\usepackage{fancyhdr}
\usepackage{enumitem}
\usepackage{balance}
\usepackage{booktabs}
\usepackage{pifont}
\usepackage{listings}
\usepackage{multirow}
\usepackage[scaled]{beramono}
\usepackage{tabularx}
\usepackage{float}

\makeatletter
\newcommand\BeraMonottfamily{%
  \def\fvm@Scale{0.85}% scales the font down
  \fontfamily{fvm}\selectfont% selects the Bera Mono font
}
\makeatother

\lstdefinestyle{SQLStyle}{
language=SQL,
basicstyle=\BeraMonottfamily\scriptsize, 
keywordstyle=\color{sqlcolor}\bfseries,
aboveskip = 0.05in,
belowskip = 0.05in,
literate = {-}{-}1, % <------ trick!
}

\lstdefinestyle{ScriStyle}{
language=SQL,
basicstyle=\BeraMonottfamily\scriptsize, 
keywordstyle=\color{smtcolor}\bfseries,
morekeywords={and, or, not},
aboveskip = 0.05in,
belowskip = 0.05in,
literate = {-}{-}1, % <------ trick!
}

\usepackage{aliascnt}

\usepackage[ruled,linesnumbered]{algorithm2e}
\usepackage{enumitem}

\SetAlFnt{\small}
\SetAlCapFnt{\small}
\SetAlCapNameFnt{\small}

\usepackage{semantic}

\usepackage{stmaryrd}
\usepackage{ltablex}
\usepackage{mathtools}

\usepackage[capitalize,noabbrev,nameinlink]{cleveref}

\newcommand{\mail}[1]{\href{mailto:#1}{#1}}

\crefname{lstlisting}{listing}{listings}
\Crefname{lstlisting}{Listing}{Listings}

\newtheorem{lemma}{Lemma}
\crefname{lemma}{Lemma}{Lemmas}
\Crefname{lemma}{Lemma}{Lemmas}

\input{cmds}

\input{std}

\vldbTitle{\paperTitle}
\vldbAuthors{\paperAuthors}
\vldbDOI{https://doi.org/10.14778/xxxxxxx.xxxxxxx}
\vldbVolume{xx}
\vldbNumber{yyy}
\vldbYear{2021}

\title{\paperTitle\vspace*{-0.15in}}
\author{\paperAuthors}

\begin{document}

\title{\paperTitle}

\author{
\def\arraystretch{1}
\begin{tabular}{ ccccc } 
 Qi Zhou  & Joy Arulraj & Shamkant Navathe & William Harris & Jinpeng Wu \\
 \multicolumn{3}{c}{\mail{equitas@cc.gatech.edu}} &
 \multicolumn{1}{c}{\mail{wrharris@galois.com}} &
 \multicolumn{1}{c}{\mail{jinpeng.wjp@alibaba-inc.com}} \\
 \multicolumn{3}{c}{\affaddr{Georgia Institute of Technology}}&
 \multicolumn{1}{c}{\affaddr{Galois.Inc}}&
 \multicolumn{1}{c}{\affaddr{Alibaba Group}}\\
\end{tabular}
}

\maketitle

\input{abstract}

\vspace{1in}

% abstract.tex: DEP, contains abstract
\input{introduction}

\input{background}
\input{overview}

\input{approach}

\input{equivalence}
\input{sound}
\input{evaluation}

\input{related-work}

\input{conclusion}
%% ==================================================================
%% BIBLIOGRAPHY
%% ==================================================================
\newpage

\bibliographystyle{abbrv}
%\small
%\raggedright
\bibliography{p}

\newpage

\appendix

\input{appendix}

\end{document}

%% file: cmds.tex
\fvset{fontsize=\scriptsize,xleftmargin=8pt,numbers=left,numbersep=5pt}

\input{code/fmt}

\def\Snospace~{\S{}}
\if 0

\fi

\newif\ifdraft\drafttrue
\newif\ifnotes\notestrue
\ifdraft\else\notesfalse\fi

\input{glyphtounicode}
\newcolumntype{R}[1]{>{\raggedleft\let\newline\\\arraybackslash\hspace{0pt}}p{#1}}

\newcommand{\PP}[1]{
\vspace{4px}
\noindent{\bf{#1}}\xspace
}

\newcommand{\squishitemize}{
 \begin{list}{$\bullet$}
  { \setlength{\itemsep}{0pt}
     \setlength{\parsep}{3pt}
     \setlength{\topsep}{3pt}
     \setlength{\partopsep}{0pt}
     \setlength{\leftmargin}{1.95em}
     \setlength{\labelwidth}{1.5em}
     \setlength{\labelsep}{0.5em} } }

\newcounter{Lcount}
\newcommand{\squishlist}{
    \begin{list}{\arabic{Lcount}. }
   { \usecounter{Lcount}
        \setlength{\itemsep}{0pt}
        \setlength{\parsep}{3pt}
        \setlength{\topsep}{3pt}
        \setlength{\partopsep}{0pt}
        \setlength{\leftmargin}{2em}
        \setlength{\labelwidth}{1.5em}
        \setlength{\labelsep}{0.5em} } }

\newcommand{\squishend}{\end{list}}

\usepackage{xstring}

\newcommand{\eg}{\textit{e.g.}}
\newcommand{\ie}{\textit{i.e.}}

\newcommand{\sys}{\textsc{SPES}\xspace}
\renewcommand{\cos}{\textsc{COSETTE}\xspace}
\newcommand{\udp}{\textsc{UDP}\xspace}
\newcommand{\eqs}{\textsc{EQUITAS}\xspace}
\newcommand{\calcite}{\textsc{Calcite}\xspace}

\newcommand{\emp}{\texttt{\textsc{EMP}}\xspace}
\newcommand{\dept}{\texttt{\textsc{DEPT}}\xspace}
\newcommand{\empid}{\texttt{\textsc{EMP\_ID}}\xspace}
\newcommand{\sal}{\texttt{\textsc{SALARY}}\xspace}
\newcommand{\loc}{\texttt{\textsc{LOCATION}}\xspace}

\newcommand{\deptid}{\texttt{\textsc{DEPT\_ID}}\xspace}
\newcommand{\deptname}{\texttt{\textsc{DEPT\_NAME}}\xspace}

\newcommand{\qa}{\texttt{Q1}\xspace}
\newcommand{\qb}{\texttt{Q2}\xspace}

\newcommand{\va}{\texttt{v1}\xspace}
\newcommand{\vb}{\texttt{v2}\xspace}
\newcommand{\vc}{\texttt{v3}\xspace}
\newcommand{\vd}{\texttt{v4}\xspace}

\newcommand{\na}{\texttt{n1}\xspace}
\newcommand{\nb}{\texttt{n2}\xspace}
\newcommand{\nc}{\texttt{n3}\xspace}
\newcommand{\nd}{\texttt{n4}\xspace}

\newcommand{\ta}{\texttt{t1}\xspace}
\newcommand{\tb}{\texttt{t2}\xspace}
\newcommand{\tc}{\texttt{t3}\xspace}
\newcommand{\td}{\texttt{t4}\xspace}

\newcommand{\nul}{\texttt{\textsc{NULL}}\xspace}

\newcommand{\tru}{\texttt{\textsc{TRUE}}\xspace}
\newcommand{\fal}{\texttt{\textsc{FALSE}}\xspace}
\newcommand{\ojoin}{\texttt{\textsc{OUTER JOIN}}\xspace}

\newcommand{\ijoin}{\texttt{\textsc{INNER JOIN}}\xspace}

\newcommand{\opand}{\texttt{\textsc{AND}}\xspace}
\newcommand{\opor}{\texttt{\textsc{OR}}\xspace}

\newcommand{\qscan}{\texttt{\textsc{SELECT}}\xspace}
\newcommand{\qproj}{\texttt{\textsc{PROJECT}}\xspace}
\newcommand{\qjoin}{\texttt{\textsc{JOIN}}\xspace}

\newcommand{\tabschema}{\texttt{\textsc{T-Schema}}\xspace}

\newcommand{\tabcol}{\texttt{\textsc{COLS}}\xspace}
\newcommand{\tabval}{\texttt{\textsc{Val}}\xspace}
\newcommand{\tabisnull}{\texttt{\textsc{Is-Null}}\xspace}

\newcommand{\cond}{\texttt{\textsc{COND}}\xspace}
\newcommand{\assign}{\texttt{\textsc{ASSIGN}}\xspace}

\newcommand{\cast}{\texttt{\textsc{CAST}}\xspace}

\newcommand{\PPS}[1]{
\noindent{\bf\textsc{#1}}
}

\newcommand{\sql}{\texttt{SQL}\xspace}
\newcommand{\SUnion}[1]{\textsc{Union}(#1)\xspace}
\newcommand{\TABLE}[1]{\textsc{Table}(#1)\xspace}
\newcommand{\SPJ}[3]{\textsc{SPJ}(#1,#2,#3)\xspace}
\newcommand{\SAggregate}[3]{\textsc{Agg}(#1,~#2,~#3)\xspace}
\newcommand{\evaluation}[3]{\langle #2 [#1] \rangle \Downarrow #3\xspace}
\newcommand{\inputs}{\texttt{Ts}\xspace}
\newcommand{\constructPred}{\textsf{ConstPred}\xspace}
\newcommand{\constructAggCall}{\textsf{CtrAgg}\xspace}
\newcommand{\fresh}{\textsf{InitTuple}\xspace}
\newcommand{\freshAgg}{\textsf{InitAgg}\xspace}
\newcommand{\expr}{\textsf{e}\xspace}
\newcommand{\exprE}{\mathsf{E}\xspace}
\newcommand{\UNF}{\textsc{unf}\xspace}
\newcommand{\SPJE}{\textsc{spje}\xspace}
\newcommand{\column}[1]{\mathsf{Column}~#1}
\newcommand{\constant}[1]{\mathsf{Const}~#1}
\newcommand{\bin}[3]{\mathsf{Bin}~#1~#2~#3}
\newcommand{\binE}[3]{\mathsf{BinE}~#1~#2~#3}
\newcommand{\binL}[3]{\mathsf{BinL}~#1~#2~#3}
\newcommand{\Function}[2]{\mathsf{Fun}~#1~(#2)}
\newcommand{\Pair}{\mathsf{Pair}~}
\newcommand{\CASE}{\mathsf{CASE}~}
\newcommand{\WHEN}[2]{(\mathsf{WHEN}~#1~#2)}
\newcommand{\conexpr}{\textsf{ConstExpr}}
\newcommand{\veriCard}{\textsf{VeriCard}}
\newcommand{\veriCardExpr}{\textsf{VeriVec}}
\newcommand{\veriCardTable}{\textsf{VeriTable}}
\newcommand{\veriCardSPJ}{\textsf{VeriSPJ}}
\newcommand{\veriCardUnion}{\textsf{VeriUnion}}
\newcommand{\veriCardAgg}{\textsf{VeriAgg}}
\newcommand{\PN}{\textsf{ConstAssign}}
\newcommand{\Compose}{\textsf{Compose}}
\newcommand{\predC}{\mathsf{P}\xspace}
\newcommand{\op}{\mathsf{OP}\xspace}
\newcommand{\cp}{\mathsf{CP}\xspace}
\newcommand{\logic}{\mathsf{LOGIC}\xspace}
\newcommand{\Not}[1]{\mathbf{Not}~#1}
\newcommand{\isNull}[1]{\mathbf{IsNull}~#1}
\newcommand{\tblcount}[2]{|#2|_{#1}}
\newtheorem{mydef}{Def}
\newtheorem{theorem}{Theorem}

\usepackage{siunitx}                                                                                                                                  
\usepackage{tikz}

\newcommand{\dcircle}[1]{\ding{\numexpr181 + #1}}

\newcounter{example}[section]
\newenvironment{example}[2][]{
\vspace{0.1in}
\refstepcounter{example}\par
\PP{Example~\theexample #1. #2:} \rmfamily}{\medskip}

\crefname{example}{Example}{Examples}

%% file: std.tex
\newcommand{\compose}{\circ}

\newcommand{\true}{\mathsf{True}}

%% file: abstract.tex
\begin{abstract}
In database-as-a-service platforms, automated verification of query equivalence
helps eliminate redundant computation in the form of overlapping sub-queries.
Researchers have proposed two pragmatic techniques to tackle this problem.
The first approach consists of reducing the queries to algebraic expressions 
and proving their equivalence using an algebraic theory.
The limitations of this technique are threefold.
It cannot prove the equivalence of queries with significant differences in the
attributes of their relational operators (\eg, predicates in the filter
operator).
It does not support certain widely-used \sql features (\eg, \nul values).
Its verification procedure is computationally intensive.
The second approach transforms this problem to a constraint satisfaction problem
and leverages a general-purpose solver to determine query equivalence.
This technique consists of deriving the symbolic representation of the queries
and proving their equivalence by determining the query containment relationship
between the symbolic expressions.
While the latter approach addresses all the limitations of the former technique, 
it only proves the equivalence of queries under set semantics 
(i.e., output tables must not contain duplicate tuples).
However, in practice, database applications use bag semantics  
(i.e., output tables may contain duplicate tuples)

In this paper, we introduce a novel symbolic approach for proving query
equivalence under bag semantics.
We transform the problem of proving query equivalence under bag semantics to
that of proving the existence of a bijective, identity map between
tuples returned by the queries on all valid inputs.
We classify \sql queries into four categories, and propose a set of
novel category-specific verification algorithms.
We implement this symbolic approach in \sys and demonstrate that \sys proves
the equivalence of a larger set of query pairs (95/232) under bag semantics 
compared to the state-of-the-art tools based on algebraic (30/232) and symbolic
approaches (67/232) under set and bag semantics, respectively.
Furthermore, \sys is 3$\times$ faster than the symbolic tool that proves
equivalence under set semantics.
\end{abstract}

%% file: introduction.tex
\section{Introduction}
\label{sec:introduction}

% why we need automaticaly decides query equivalence
Database-as-a-service (DBaaS) platforms enable users to quickly deploy complex
data processing pipelines consisting of SQL
queries~\cite{maxcompute,azure,bigquery}. 
These pipelines often exhibit a significant amount of computational overlap
(\ie, semantically equivalent sub-queries)~\cite{alekh18,zhou19}.
This results in higher resource usage and longer query execution times.
Researchers have developed techniques for minimizing redundant computation by
materializing the overlapping sub-queries as views and rewriting the original
queries to operate on these materialized views~\cite{rachel00,goldstein01}.
All of these techniques rely on an effective and efficient algorithm for
automatically deciding the equivalence of a pair of SQL queries.
Two queries are \textit{equivalent} if they always return the same output table
for any given input tables.

% %
% In general, proving query equivalence (QE) is an undecidable
% problem~\cite{abiteboul95,ba50}.
% %
% Given this constraint, prior efforts have focused on a subset of SQL
% queries where this problem is decidable (\eg, \qscan-\qproj-\qjoin
% queries)~\cite{ashok77,tannen99,cohen99,jayram2006}.
% %
% While this line of research has studied the theoretical underpinnings of this
% problem, these techniques are unable to identify overlap in complex SQL queries.

\PP{Prior Work:}
Although proving query equivalence (QE) is an undecidable
problem~\cite{abiteboul95,ba50}, researchers have recently formulated two
pragmatic approaches for automatically proving QE.
The first approach is based on an algebraic
representation of queries~\cite{chu2018axiomatic, chu2017sig}.
\udp, the state-of-the-art prover based on this algebraic
approach, determines QE using three steps.
First, it transforms the queries from an abstract syntax tree (AST)
representation to an algebraic representation.
Next, it applies a set of normalization rules to homogenize the algebraic
representations (ARs) of these queries.
Lastly, it attempts to find an isomorphism between the tuple vocabularies
of two normalized ARs to determine their syntactical equivalence using
substitution.	
This pragmatic algebraic approach works well on real-world \sql queries.
However, it suffers from three limitations.
First, \udp cannot prove the equivalence of queries when the attributes in their
relational operators are syntactically different (\eg, predicates in the
filter operator).
This is because it relies on a set of pre-defined syntax-driven rewrite rules to normalize 
the algebraic representation.
Second, it does not support certain widely used \sql features (\eg, \nul
values).
Third, its verification procedure is computationally intensive due to the 
large number of possible normalized ARs.

\eqs addresses these limitations by reducing the problem of proving QE to a
constraint satisfaction problem~\cite{zhou19}. 
\eqs determines QE using two steps.
First, it transforms the queries from an AST representation to a symbolic
representation (SR) (\ie, a set of first-order logic (FOL) formulae).
A query's SR symbolically represents an arbitrary tuple that it returns.
Next, it leverages a general-purpose solver based on satisfiability modulo
theory (SMT) to determine the containment relationship between two
SRs\footnote{If every tuple returned by a query $Y$ is returned by query $X$ on
all possible input databases, then $X$ contains $Y$. If $X$ contains $Y$ and 
vice versa, then $X$ and $Y$ are equivalent.}.
If the containment relationship holds in both directions, then the queries are
equivalent.
\eqs is empirically more effective and efficient than \udp~\cite{zhou19}.

While this symbolic approach addresses the drawbacks of the algebraic approach,
it has one significant limitation.
It only proves the equivalence of queries under \textit{set} 
semantics (\ie, output tables must not contain duplicate
tuples~\cite{Negri1991}).
In practice, database applications use \textit{bag} semantics 
(\ie, output tables may contain duplicate tuples~\cite{Albert1991})
Proving QE under bag semantics is a strictly harder problem than doing so
under set semantics.
This is because if two queries are equivalent under bag semantics, then
they are also equivalent under set semantics.
However, the converse does not hold.
It is \textit{infeasible} to extend \eqs to prove query equivalence under bag
semantics.
To prove QE under set semantics, it verifies the containment relationship
between queries.
This technique only works if all the tuples in the output tables of the queries
are distinct.
So, it does not work if duplicate tuples are present in the tables.

\PP{Our Approach:}
In this paper, we present a novel technique for proving QE under bag semantics.
We reduce the problem of proving that two queries are equivalent under bag
semantics to the problem of proving the existence of a bijective, identity map between tuples returned
by the queries across all valid inputs.
We introduce a novel query pair symbolic representation (QPSR) to symbolically
represent the bijective map between tuples returned by two \textit{cardinally
equivalent} queries and leverage an SMT solver to efficiently verify that the
bijective map is an identity map
\footnote{Two queries are cardinally equivalent if and only if they return the
same number of tuples for all valid inputs.}.
We classify \sql queries into four categories, and propose a set of novel 
category-specific verification algorithms to verify cardinal equivalence and 
construct the QPSR.

We implemented this symbolic approach in \sys, a tool for automatically
verifying the equivalence of SQL queries under bag semantics.
We evaluate \sys using a collection of 232 pairs of equivalent SQL queries
available in the Apache \calcite framework~\cite{calcite}.
Each query pair is constructed by applying various optimization rules on complex
\sql queries with diverse features (\eg, arithmetic operations, three-valued
logic for supporting \nul, sub-queries, grouping, and aggregate functions).
Our evaluation shows that \sys proves the semantic equivalence of a larger
set of query pairs (95/232) compared to \udp (34/232) and \eqs (67/232).
Furthermore, \sys is 3$\times$ faster than \eqs on this benchmark.
In addition to the \calcite benchmark, we evaluate the efficacy of \sys on a
cloud-scale workload comprising $9,486$ real-world SQL queries from Ant
Financial Services Group~\cite{antfin}.
\sys automatically found that 27\% of the queries in this workload contain
overlapping computation.
Furthermore, 48\% overlapping queries contain either aggregate or join operations, 
which are expensive to evaluate.

\PP{Contributions:}
We make the following contributions:
\squishitemize 
%
%\item 
%We highlight the limitation of the previous symbolic approach to 
%determining the equivalence of \sql queries in~\autoref{sec:background}.
%
\item 
We present a novel query pair symbolic representation and a set of verification
algorithms for determining QE under bag semantics
in~\autoref{sec:equivalence}.
\item 
We classify \sql queries into four categories, and present a set of 
normalization rules to enhance the ability of \sys to prove query equivalence
under bag semantics in~\autoref{sec:approach}. 
\item We implement this approach in \sys and evaluate its efficacy.
We demonstrate that \sys proves the equivalence of a larger set of query pairs
in \calcite benchmark compared to the state-of-the-art tools
in~\autoref{sec:evaluation}.
More importantly, unlike \eqs, \sys proves QE under bag semantics.

\squishend

\vspace{5px}

%% file: background.tex
\section{Background \& Motivation}
\label{sec:background}
In this section, we present an overview of the symbolic approach for proving QE 
under set semantics used in \eqs~\cite{zhou19}.
We then explain why this approach cannot be used to prove QE under
bag semantics.

A pair of queries \qa and \qb are equivalent under set semantics if and only if
for all valid input tables, \qa and \qb return the same set of tuples.
%
%Under set semantics, a table cannot contain duplicate tuples. 
%
\eqs proves QE by verifying whether the \textit{containment} relationship holds
in both directions (\ie, $\qa$ contains $\qb$ and $\qb$ contains $\qa$). 
$\qa$ is contained by $\qb$ under set semantics if and only if the following
condition is satisfied:
for all valid input tables, if any tuple is returned by $\qa$, then there is
always a corresponding, identical tuple returned by $\qb$.
\eqs first constructs the SRs of \qa and \qb, and then uses an SMT solver to
verify the relational properties between two SRs to prove the containment
relationship.
We present an example to demonstrate how \eqs proves QE.
The queries are based on these two tables:
\squishitemize 
\item \emp table: $\langle$\empid, \sal, \deptid,\loc$\rangle$ 
\item \dept table: $\langle$\deptid, \deptname$\rangle$
\squishend

\begin{example}{Set Semantics}
\label{ex:set}
\begin{lstlisting}[style=SQLStyle] 
Q1: SELECT EMP.DEPT_ID, EMP.LOCATION FROM EMP WHERE DEPT_ID > 10;
Q2: SELECT EMP.DEPT_ID, EMP.LOCATION FROM EMP
      GROUP BY EMP.DEPT_ID, EMP.LOCATION ON DEPT_ID + 5 > 15;
\end{lstlisting}
\end{example}

\qa selects the department id and location columns of all employees whose
department id is greater than 10.
\qb first selects the same columns grouped by the department id and location
columns. 
\qb then filters out tuples whose department id plus five is not greater than
15.
\qa and \qb are equivalent under set semantics because they return the same set
of tuples. 

\eqs first constructs the SRs for \qa and \qb:
\begin{lstlisting}[style=ScriStyle] 
Q1: Cond1: (v3 > 10 and !n3); Cols1: {(v3,n3),(v4,n4)}
Q2: Cond2: (v3 + 5 > 15 and !n3); Cols2: {(v3,n3),(v4,n4)}
\end{lstlisting}
Each SR contains two parts: $\cond$ and $\vec{\tabcol}$.
$\cond$ is an SMT formula that each tuple needs to satisfy (\ie, the condition)
in order to be returned by the given query. 
For example, in $\qa$'s SR, $\cond$ is $v3$ > $10$ and $!n3$, which means
that each returned tuple's department id must be greater than ten and must 
not be $\nul$.
$\vec{\tabcol}$ is a vector of pairs of SMT formula that symbolically 
represents an arbitrary tuple returned by the query.
Each pair of SMT formula represents a column, where the first formula 
represents the value of the column and the second boolean formula indicates 
if the value is $\nul$.
For example, in $\qa$'s SR, $\vec{\tabcol}$ has two pairs of SMT formulae,
where the first pair of SMT formulae symbolically represents the department id 
column, and the second pair of SMT formulae represents the location column.

\eqs leverages the SMT solver to verify \textit{two} relational properties
between the SRs to prove that \qa is contained by \qb.
It first verifies that $\cond_1 \implies \cond_2$.
\eqs checks this property to show that if there is an unique tuple returned by
\qa, then there is always a corresponding unique tuple returned by \qb.
It then verifies if $(\cond_1 \land \cond_2 ) \implies (\vec{\tabcol_1} =
 \vec{\tabcol_2})$.
\eqs checks this property to show that the two unique tuples returned by \qa and
\qb are always equivalent.
If these two properties hold, then \qa is contained by \qb.
\eqs uses the same technique to prove that \qb is contained by \qa.
Taken together, these containment relationships mean than \qa and \qb are
equivalent.

\PP{Bag Semantics:}
\begin{figure}[t]
  \centering
  \includegraphics[width=0.85\linewidth]{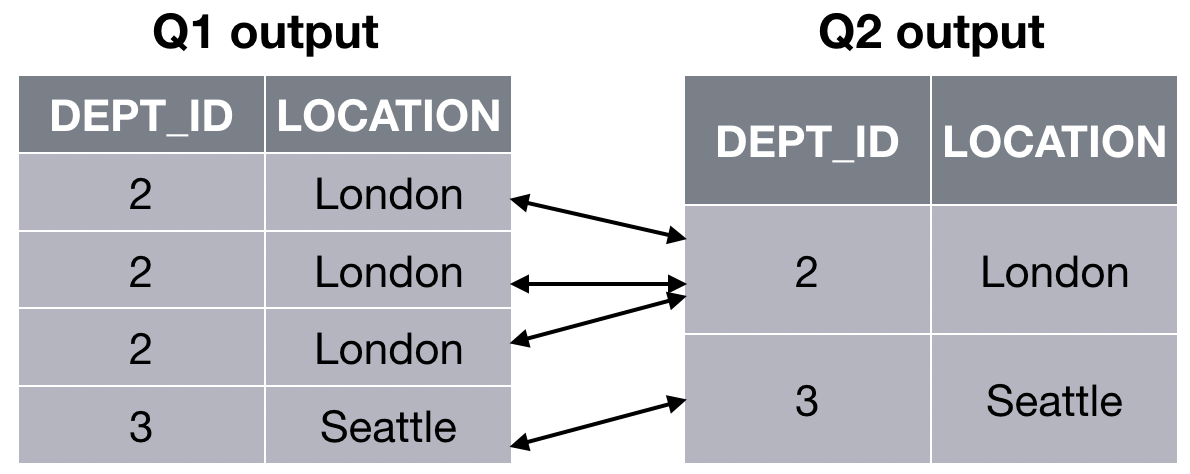}
  \caption{
  \textbf{Limitations of \eqs}  --
  Output tables are equivalent under set semantics (but not under bag
  semantics).}
  \label{fig:setfail}
\end{figure}
\qa and \qb are \textit{not} equivalent under bag semantics. 
Under bag semantics, a table may contain duplicate tuples.
\qa and \qb would return different output tables if there were two employees 
working in the same department and location.
As shown in \cref{fig:setfail}, three employees working in the same department
and location show up in the return table of \qa.
But, they only show up once in the return table of \qb.
Since database applications tend to assume bag semantics, it is critical to
prove QE under bag semantics.

It is not feasible to generalize the containment approach used in \eqs to support
bag semantics.
\eqs derives the SR of \textit{one} arbitrary tuple in the returned table.
It proves the containment relationship by proving there is an identical tuple in
the output table of the other query. 
This is sufficient for proving QE under set semantics because each tuple occurs
only once in the output table.
However, it is not sufficient for proving QE under bag semantics,
since each distinct tuple may appear multiple times in the output table.
It cannot track and compare the frequencies of occurrence of these tuples.
As shown in \cref{fig:setfail}, the first three tuples in \qa's output table
are proved to be equal to one tuple in \qb's output.  
So, \eqs incorrectly concludes that \qa is contained by \qb under bag semantics.

It is not feasible to construct an SR with an unbounded number of symbolic
tuples, where each bag of tuples represented by the SR is different.
As shown in \cref{fig:setfail}, the maximum number of times a tuple appears in
\qa's output table may be arbitrarily large.
It is infeasible to construct an SR with a fixed number of symbolic tuples such
that each bag of tuples represented by the SR is guaranteed to be different.
Therefore, proving that two queries are equivalent under bag semantics requires
a novel approach.

%% file: overview.tex
\section{Overview}
\label{sec:overview}

In this section, we first present an overview of how \sys proves QE 
under bag semantics in~\autoref{sec:overview::high-level}.
We then illustrate this approach using an example
in~\autoref{sec:overview::example}.

\subsection{Problem Formulation}
\label{sec:overview::high-level}

To prove QE under bag semantics, we present a novel problem formulation.
Based on the definition of bag semantics, two queries are equivalent under 
bag semantics if and only if for all valid input tables, 
both queries return the same multi-valued set of tuples.
So, for all valid input tables, if there always exists a \textit{bijective,
identity map} between the output tables, then the two queries are equivalent.
As shown in~\cref{fig:fulleq}, a bijective, identity map is a one-to-one map 
that maps a tuple in one output table to an unique, identical tuple in the
other output table.
Even if the output table has duplicate tuples, each of them would show up once
and exactly once in the bijective, identity map.
Thus, we transform the problem of proving QE under bag semantics to proving the
existence of a bijective, identity map for all valid input tables.

\sys tackles the latter problem into two steps.
\begin{figure}[t]
   \centering
   \subfloat[Cardinal Equivalence\label{fig:cardeq}]{%
          \includegraphics[width=0.45\linewidth]{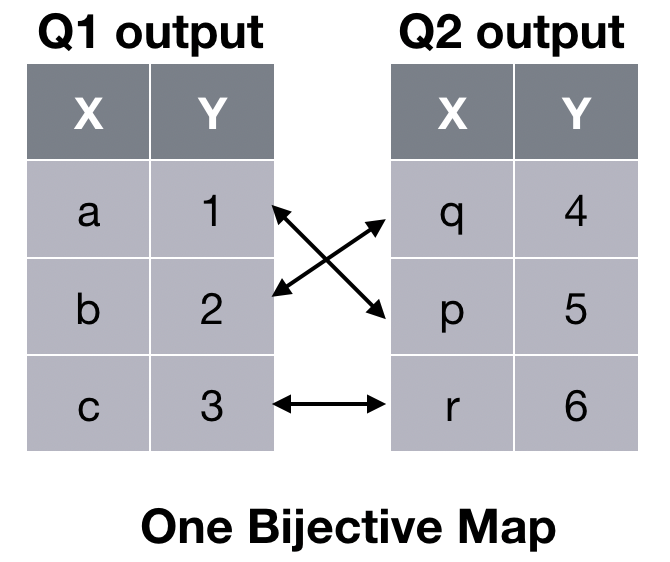}
   }
   \subfloat[Full Equivalence\label{fig:fulleq}]{%
          \includegraphics[width=0.46\linewidth]{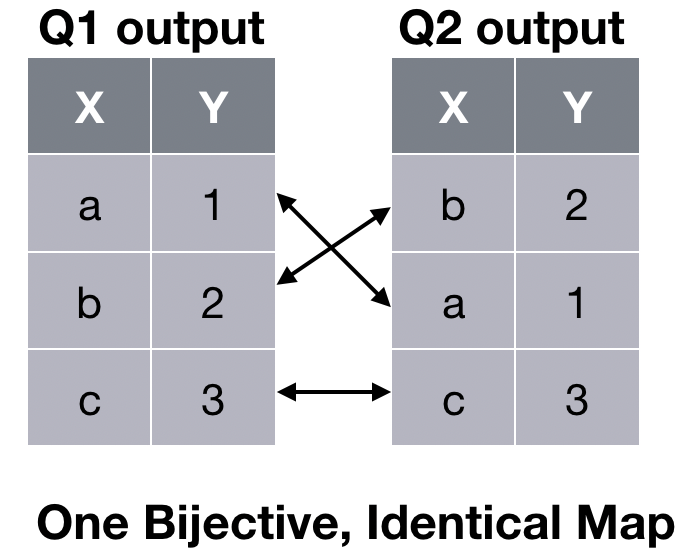}
   }
  \caption{
  \textbf{Types of Query Equivalence --} 
  Bijective maps implicitly constructed by \sys to
  determine: (a) cardinal equivalence but not full equivalence and (b) full equivalence of 
  queries under bag semantics.}  
  \label{fig:eval:sched:sys}
\end{figure}

\PP{\dcircle{1} Cardinal Equivalence:}
In the first step, \sys proves that for all valid inputs, there exists a
bijective map between tuples in the output tables of two given queries.
\sys first verifies if the given pair of queries are 
\textit{cardinally equivalent} under bag semantics to verify the existence of a
bijective map. 
Two queries are cardinally equivalent if and only if for all valid inputs, 
their output tables contain the same number of tuples.
We defer a formal definition of cardinal equivalence
to~\autoref{sec:equivalence::definitions}.
If two queries are cardinally equivalent, then there exists a \textit{bijective
map} between the tuples returned by these two queries for all valid inputs, 
as shown in~\cref{fig:cardeq}.
In this map, each tuple in the first table is mapped to a unique tuple in the
second table, and all tuples in second table are covered by the map.
We note that the contents of the output tables of two cardinally equivalent
queries may not be identical (\ie, they are not fully equivalent). 

\sys then constructs a \emph{Query Pair Symbolic Representation (QPSR)} for two
cardinally equivalent queries to symbolically represent the bijective map between
the returned tuples for all valid inputs.
The QPSR contains a pair of symbolic tuples to represent the bijective map.
The first tuple represents an arbitrary tuple returned by the first query.
The second tuple represents the corresponding tuple returned by the second query
as defined by the bijective map.
Given a set of concrete input tables, each pair of tuples in the bijective map
that are returned by two queries can be obtained by substituting the two symbolic
tuples with a \textit{model} (\ie, a set of concrete values for the symbolic
variables).
We defer a discussion of QPSR to~\autoref{sec:equivalence::check}.

\sys proves the cardinal equivalence of two queries by recursively constructing 
the QPSR of their sub-queries and using the SMT solver to verify specific
properties of construed sub-QPSRs based on the semantics of different types of
queries.
We defer a discussion of how \sys proves cardinal equivalence and constructs QPSR
to~\cref{sec:equivalence::construct,sec:equivalence::construct:::table,sec:equivalence::construct:::spj,sec:equivalence::construct:::agg,sec:equivalence::construct:::union}.

\PP{\dcircle{2} Full Equivalence:}
In the second step, \sys proves that the symbolic bijective map is always an
\textit{identity} map to prove that the given pair of queries are \textit{fully
equivalent} under bag semantics (\ie, their contents are identical as well). 
Two queries are fully equivalent if and only if for all valid input tables, 
their output tables contain the same tuples (ignoring the order of the tuples).
We defer a formal definition of full equivalence
to~\autoref{sec:equivalence::definitions}.
If there exists a \textit{bijective,
identity map} between the tuples returned by these two queries for all valid
inputs, then the two queries are fully equivalent.
\cref{fig:fulleq} shows an example of one bijective, identity map for one
possible set of output tables. 
In this map, each tuple in the first table is mapped to a unique, 
identical tuple in the second table.
All tuples in the second table are covered by this map.

\sys uses the constructed QPSR to verify that the given pair of queries are
fully equivalent.
It proves full equivalence by showing that the bijective map is an identity map.
Because the QPSR uses a pair of symbolic tuples to represent the bijective
tuple, \sys leverages the SMT solver to prove that the two symbolic tuples are
always equivalent for all valid \textit{models}.  
We defer a discussion of how to prove the full equivalence using QPSR to
\autoref{sec:equivalence::check}.

\PP{\sys vs \eqs:}
The key differences between \sys and \eqs lies in how they prove QE.
\squishitemize
\item
\sys converts the QE problem to a problem of proving the existence of a
bijective, identity map between tuples in the output tables of the two queries
for all valid input tables.
\eqs reduces the QE problem to a query containment problem by showing that an
arbitrary tuple returned by one query is also returned by the other query.

\item
\sys constructs a QPSR for a pair of queries after verifying that
the queries are cardinally equivalent.
The QPSR symbolically represents the bijective map between the tuples in their output tables. 
In contrast, \eqs separately constructs an SR for each individual query
that represents the tuples in its output table.

\item
\sys decomposes the problem of proving equivalence of queries
into smaller proofs of equivalence of their sub-queries (\ie sub-QPSRs).
It constructs the bijective map between tuples in the final output tables by
\textit{recursively} constructing the bijective maps between tuples in all of
the intermediate output tables.  
\eqs directly proves the containment relationship for the entire queries.
So, \sys scales well to larger queries.
\squishend

\PP{SMT Solver:}
\sys leverages an SMT solver to prove cardinal and full equivalence of queries~\cite{moura08}.
An SMT solver determines if a given FOL formula is satisfiable.
For example, the solver decides that the following formula can be satisfied: 
$ x+5 > 10 \land x > 3$ when $x$ is six.
Similarly, it determines that the following formula cannot be satisfied: 
$ x+5 > 10 \land x < 4$ since there is no integral value of $x$ for which this
formula holds.
A detailed description of solvers is available in~\cite{DeMoura2011}.

\subsection{Illustrative Example}
\label{sec:overview::example}
We now use an example to show how \sys proves QE under bag semantics.
These queries are based on the \emp and \dept tables.
\begin{example}{Bag Semantics}
\label{ex:aggregate}
\begin{lstlisting}[style=SQLStyle] 
Q1: SELECT SUM(T.SALARY), T.LOCATION FROM
     (SELECT SALARY, LOCATION FROM DEPT, EMP 
      WHERE EMP.DEPT_ID = DEPT.DEPT_ID AND DEPT_ID + 5 = 15) AS T 
    GROUP BY T.LOCATION;
Q2: SELECT SUM(T.SALARY), T.LOCATION FROM
     (SELECT SALARY, LOCATION, DEPT_ID FROM EMP, DEPT 
       WHERE EMP.DEPT_ID = DEPT.DEPT_ID AND DEPT_ID = 10) AS T 
    GROUP BY T.LOCATION, T.DEPT_ID;
\end{lstlisting}
\end{example}

\qa is an aggregation query that calculates the sum of salaries of employees 
in the same location, whose department id plus five is $15$.
\qb is an aggregation query that calculates the sum of salaries of employees 
in the same location and department, whose department id is $10$.
\qa and \qb are semantically equivalent under bag semantics, 
since the inner query retrieves tuples whose department id is 10 (and 10 + 5 =
15), and the two \texttt{GROUP BY} sets cluster tuples in the same manner
(department id is constant). 

\begin{figure}[t!]
\centering
\includegraphics[width=0.8\linewidth]{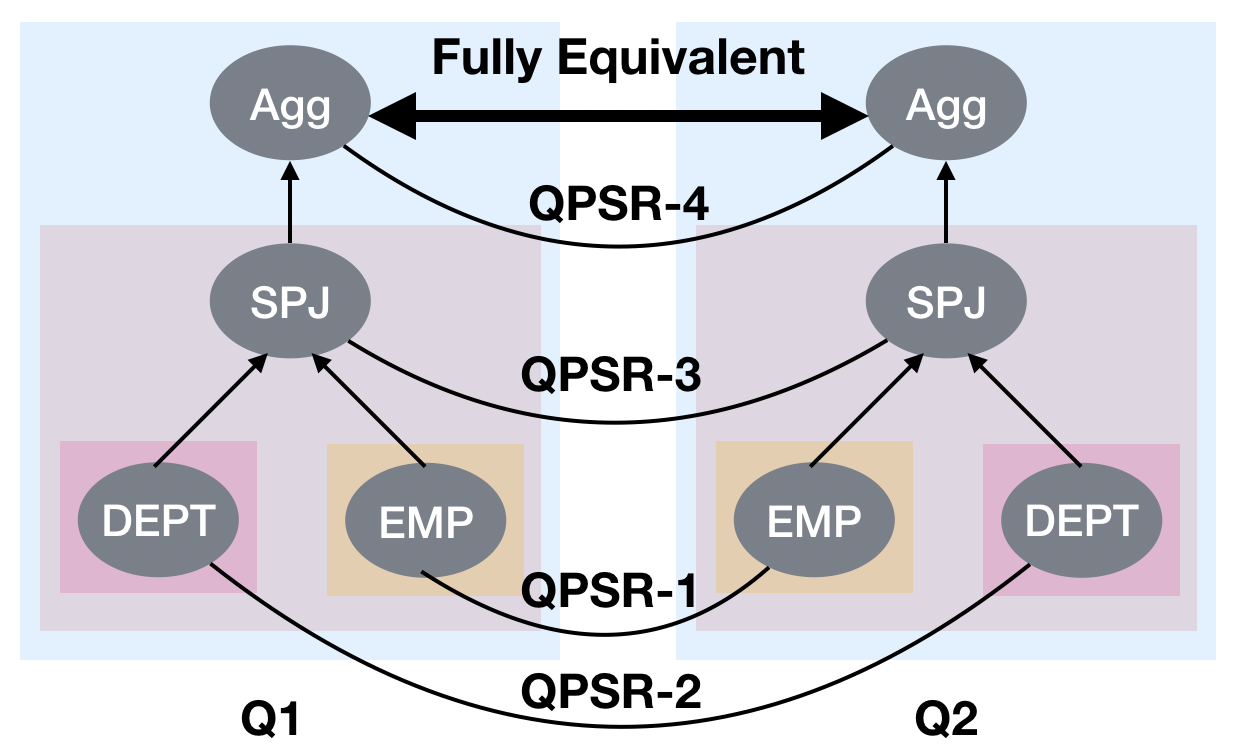}
\caption{
 \textbf{Illustrative Example} -- The two-stage approach that \sys uses to
 prove QE under bag semantics.}
\label{fig:ar}
\end{figure}

\PP{Tree Representation:}
\sys first normalizes each queries into a tree representation.
This representation is similar to a logical execution plan, which captures 
the semantics of the original queries.
The only difference is that each node in this tree represents a 
sub-query that belongs to a specific category (\eg, SPJ query).
We classify \sql queries into four categories based on their constructors.
We defer a discussion of these four categories to~\autoref{sec:approach}.

\autoref{fig:ar} shows the tree representation of queries \qa and \qb.
The tree representation of \qa is an aggregate query that takes a
\qscan-\qproj-\qjoin (SPJ) sub-query as input, 
the department id as the group set, and the sum of salaries as the aggregate
operation. 
The SPJ node takes two table sub-queries (\texttt{EMP} and
\texttt{DEPT}) as inputs, and has a filter predicate on \texttt{DEPT}.
\sys constructs the tree representation of \qb in a similar manner. 

\PP{Proving QE Under Bag Semantics:}
To prove \qa and \qb are equivalent under bag semantics, 
\sys first verifies the cardinal equivalence of two queries.
In order to verify the cardinal equivalence of two aggregate queries, 
\sys recursively constructs the QPSR of two SPJ sub-queries that 
the aggregate queries take as inputs.
To verify the cardinal equivalence of two SPJ sub-queries, 
\sys constructs a bijective map between their input sub-queries and checks if
each pair of input sub-queries are cardinally equivalent.
If that is the case, then it constructs a QPSR for each pair of input
sub-queries. 
\sys pairs the \texttt{EMP} and the \texttt{DEPT} tables in \qa with those in
\qb, respectively.

\PP{QPSR-1:}
The QPSR for the pair of \texttt{EMP} tables is given by:
\begin{lstlisting}[style=ScriStyle] 
COND: True
COLS1: {(v1,n1),(v2,n2),(v3,n3),(v4,n4)} 
COLS2: {(v1,n1),(v2,n2),(v3,n3),(v4,n4)}
\end{lstlisting}
\label{ex:table::qpsr}

Here, $\tabcol_1$ and $\tabcol_2$ symbolically represent two corresponding
tuples present in the two cardinally equivalent tables, respectively.
Each symbolic tuple is a vector of pairs of FOL terms. 
We present the formal definitions of $\tabcol_1$ and $\tabcol_2$
in~\autoref{sec:equivalence::check}.
This pair of symbolic tuples
$\tabcol_1$ and $\tabcol_2$ defines a bijective map between the
tuples returned by the table.
Since both tables refer to \texttt{EMP}, the bijective map is an identity map.

$\{(\va,\na), (\vb,\nb),(\vc,\nc),(\vd,\nd)\}$ symbolically represents a tuple
returned by the \texttt{EMP} table.
Each pair of symbolic variables represents a column.
For instance, $(\va,\na)$ denotes \empid in this symbolic tuple.
$\va$ represents the value of EMP_ID, $\na$ indicates if the value is NULL.
The encoding scheme is similar to the one used by \eqs~\cite{zhou19}.
$\cond$ is an FOL formula that represents the filter conditions.
It must be satisfied for the tuples to be present in the output table of this query.
$\cond$ is \tru because the \textit{table query} (\autoref{sec:approach::ar})
simply returns all tuples present in the table.

\PP{QPSR-2:}
The QPSR for the pair of \texttt{DEPT} tables is given by:
\begin{lstlisting}[style=ScriStyle] 
COND: True 
COLS1: {(v5,n5),(v6,n6)}
COLS2: {(v5,n5),(v6,n6)}
\end{lstlisting}
Here, $\{(v5,n5), (v6,n6)\}$ symbolically represents a tuple returned by the
\texttt{DEPT} table.

\PP{QPSR-3:}
\sys uses QPSR-1 and QPSR-2 to construct two symbolic predicates in two SPJ
queries. 
It then leverages the SMT solver to verify that two predicates always
return the same boolean results for pairs of corresponding tuples in the 
bijective map between two join tables.
In this way, it proves that the two SPJ queries are cardinally equivalent.
It then constructs a QPSR for the SPJ queries:
\begin{lstlisting}[style=ScriStyle] 
COND: (v3 + 5 = 15 and !n3) and (v3 = 10 and !n3)
COLS1: {(v1,n1),(v2,n2),(v3,n3),(v4,n4),(v5,n5),(v6,n6)} 
COLS2: {(v1,n1),(v2,n2),(v3,n3),(v4,n4),(v5,n5),(v6,n6)}
\end{lstlisting}
$\tabcol_1$ and $\tabcol_2$ symbolically represent a bijective map between tuples in the output tables of two SPJ queries. 
This bijective map \textit{preserves} the two bijective maps in the two
sub-QPSRs between their input table queries.
In other words, if a tuple $\ta$ is mapped to another tuple $\tb$ in QPSR-1, 
and a tuple $\tc$ is mapped to another tuple $\td$ in QPSR-2,
then the join tuple of $\ta$ and $\tb$ maps to that of $\tc$ and $\td$ in QPSR-3.
In this manner, the mapping in the lower-level QPSRs is preserved in the
higher-level QPSR.
$\cond$ is obtained by combining the filter predicates.

\PP{QPSR-4:}
\sys uses QPSR-3 and the SMT solver to verify that two aggregate queries always return the same number of groups of tuples based on the group by set
to prove the two aggregate queries are cardinally equivalent.
If and only if they are cardinally equivalent, 
it constructs a QPSR for the entire aggregate queries (\ie, \qa and \qb):
\begin{lstlisting}[style=ScriStyle] 
COND: (v3 + 5 = 15 and !n3) and (v3 = 10 and !n3)
COLS1: {(v1,n1),(v7,n7)} 
COLS2: {(v1,n1),(v7,n7)}
\end{lstlisting}
Here, $\tabcol_1$ and $\tabcol_2$ symbolically represent a bijective map
between tuples returned by \qa and \qb, respectively.
$(v7,n7)$ represents the sum of salaries column.

\PP{Full Equivalence:}
After determining cardinal equivalence of \qa and \qb, 
\sys uses QPSR-4 to verify the full equivalence of these queries 
by showing that the bijective map is an identity map.
It uses an SMT solver to verify the following property of QPSR-4: 
$\cond \implies \tabcol_1 = \tabcol_2$.
It feeds the negation of the property to the solver.
The solver determines that it cannot be satisfied, thereby showing that the
paired symbolic tuples are always equivalent when $\cond$ holds.
Thus, the bijective map between the tuples returned by the queries is an identity
map.
So, \qa and \qb are fully equivalent under bag semantics.

\PP{Summary:}
\sys first constructs QPSR-1 for \emp table and
QPSR-2 for \dept table. 
It then uses these QPSRs to determine the cardinal equivalence of SPJ queries.
Next, it constructs QPSR-3 for the SPJ queries.
\sys then uses QPSR-3 to determine the cardinal equivalence of aggregate queries.
Lastly, it uses QPSR-4 to check the full equivalence of \qa and \qb.
Thus, \sys always checks for cardinal equivalence before constructing the QPSRs.
It only checks the full equivalence property for the top-level QPSR (\ie,
QPSR-4).
This is because two queries may be fully equivalent even if their sub-queries
are not fully equivalent (\eg, due to aggregation)

%% file: approach.tex
\section{Query Representation}
\label{sec:approach}
In this section, we first describe the syntax and semantics of four different categories of queries in
~\autoref{sec:approach::ar}. 
\sys uses a a tree representation for processing queries.
We then introduce the \textit{Union Normal Form} (UNF) of a normalized query 
representation and a minimal set of normalization rules
in~\autoref{sec:approach::unf}.
These normalization rules enhance the ability of \sys in proving QE for 
a larger set of queries  under bag semantics.

\subsection{Syntax and Semantics}
\label{sec:approach::ar}
We first present the syntax of the four categories.
The reasons for defining these categories are twofold.
First, it allows \sys to cover most of the frequently-observed SQL queries.
Second, since different types of queries have different semantics, \sys can 
leverage category-specific QE verification algorithms.
%
% The formal definition of the semantics is given
% in~\autoref{app:ar-semantics}~\cite{zhou2020} \footnote{The appendix is included in the supplementary material.}.
%
A query $\expr$ is defined as:
\input{fig/Expr.tex}
In \sys, a query can be: %
(1) a table query, %
(2) an SPJ query, %
(3) an aggregate query, or %
(4) a union query.
Except for the table query, all other types of queries take sub-queries as
its inputs.
Thus, \sys represents the whole query as a tree, where each node is a sub-query
that belongs to one of the four categories.
We consider a table to be a bag (\ie, multi-valued set) of tuples as it best
represents real-world databases.
\sys supports the \texttt{DISTINCT} keyword for discarding duplicate tuples in 
a bag.
So, it also supports set semantics.

We next describe the semantics of the four categories  
based on the relationships between the input and output tables. 

\PP{\dcircle{1} Table Query:}
$\TABLE{n}$ represents a table.
It contains only one field: the name of the table ($n$).
This type of query returns all the tuples in table $n$.

\PP{\dcircle{2} SPJ Query:}
This type of query contains three fields: 
(1) a vector of input sub-queries ($\vec{\expr}$), 
(2) a predicate that determines whether a tuple is selected ($\textsc{p}$),
and 
(3) a vector of projection expressions that transform each selected tuple
($\vec{o}$).
A predicate may contain arithmetic operators, logical operators, and functions
that check if a term is \nul.
\sys supports higher-order predicates (\eg, \texttt{EXISTS}) which are encoded
as uninterpreted functions.
%
%
% \PP{Projection Expression:}
%
A projection expression may contain columns, constant values, \nul, arithmetic
operations, user-defined functions, and the \texttt{CASE} keyword.
The SPJ query first selects the tuples in the cartesian product of the vector of
input sub-queries that satisfy the predicate $p$, and it then emits the transformed tuples
obtained by applying the vector of projection expressions on each tuple.

% Given a set of valid input tables $\inputs$, the SPJ AR first evaluates the
% vector of input ARs on $\inputs$ to obtain a vector of input tables.
% %
% For each tuple $t$ in the \textit{cartesian product} of the vector of input
% tables, if $t$ satisfies the given predicate $p$, 
% it then applies the vector of expressions $\vec{o}$ on $t$
% and emits the transformed tuple.

% \PP{Predicate:}
%

\PP{\dcircle{3} Aggregate Query:}
The aggregate query contains three fields:
(1) an input sub-query ($\expr$), 
(2) a group by set ($\vec{g}$), and 
(3) a vector of aggregate functions ($\vec{agg}$).
The aggregate query first groups the tuples returned by the input sub-query
based on the set of columns in the group by set, such that 
tuples in each group take the same values for the columns in $\vec{g}$.
It then applies the vector of aggregate functions to each group of tuples, and
returns one aggregate tuple per group.
Each aggregate function generates a column in the tuple emitted by this query.

% Given a set of valid input tables $\inputs$, the aggregate AR first evaluates
% the input AR on $\inputs$ to get an input table $T_0$.
% %
% It then partitions the input table $T_0$ into a set of bags of tuples as defined 
% by a set of grouping attributes $\vec{g}$ 
% (tuples in each bag take the same values for the grouping attributes).
% %
% Lastly, for each bag of tuples, it applies the vector of aggregate functions
% and returns one tuple.
% %
% Each aggregate function generates a column in that tuple.

\PP{\dcircle{4} Union Query:}
The union query contains one field: a vector of input sub-queries ($\vec{\expr}$).
It returns all the tuples present in the input tables (without discarding
duplicate tuples).
The union query captures the semantics of the \texttt{UNION ALL}
operator~\cite{Guagliardo2017}. 
% %
% Given a set of valid input tables $\inputs$, the union AR first evaluates
% the vector of input ARs on $\inputs$ to get a vector of input tables.
% %
% It then returns all the tuples present in the input tables (without discarding
% duplicate tuples).
% %
% The union AR captures the semantic of the \texttt{UNION ALL}
% operator~\cite{Guagliardo2017}. 

%\PP{Motivation:}
%
%
% We include aggregate and union ARs in our definitions since they are 
% widely used SQL constructs, and their semantics differs from that of SPJ
% queries.

\PP{Complex SQL Constructs:}
\sys also supports certain complex SQL constructs that do not directly map to
these four categories by reducing them to a combination of these categories.
Here are two illustrative examples:

\PP{\dcircle{1}} 
\sys expresses the \texttt{LEFT OUTER JOIN} operator as a
\texttt{UNION} query that takes a vector of two SPJ sub-queries as input.
The first SPJ sub-query represents the \texttt{INNER JOIN} component of the
\texttt{LEFT OUTER JOIN} operator.
The second SPJ sub-query represents the \texttt{OUTER JOIN} component of the
\texttt{LEFT OUTER JOIN} operator and uses \texttt{EXISTS} in the predicate.

\PP{\dcircle{2}}
\sys expresses the \texttt{DISTINCT} operator as an aggregate query
 where the \texttt{GROUP BY} set contains all columns.

\subsection{Normalization Rules}
\label{sec:approach::unf}
The syntax of an UNF query is defined as follows:
\input{fig/Expr2.tex}
The UNF query is a union query that takes a vector of
normalized SPJ sub-queries as input ($\vec{\SPJE}$).
Each normalized SPJ query takes a vector of sub-queries as input ($\vec{\exprE}$).
These sub-queries are either a table query or a normalized
aggregate queries.
Each normalized aggregate queries can only take a UNF query as input.

\PP{Conversion to UNF:}
An query can be normalized to UNF query by repeatedly applying a set of normalization
rules.
The number of rule applications is finite and the rules are not applied in a
specific order.
Due to space constraints, we only present the most significant rule.
If a query is $\SPJ{\expr_0::\vec{\expr_1}}{p_1}{\vec{o_1}}$ and $\expr_0 =
\SPJ{\vec{\expr_2}}{p_2}{\vec{o_2}}$,  
then \sys transforms the query to $\SPJ{{\vec{\expr_2}::\vec{\expr_1}}}{p_1 \land
p_2}{\vec{o_1} \circ \vec{o_2}}$.
Here, $::$ denotes concatenation of two vectors and $\circ$ represents
element-wise composition of two vectors of projection expressions.
There are two additional rules to merge Union and SPJ queries, respectively.

\PP{Normalization Rules:}
\sys uses a minimal set of pre-defined rules to further simplify 
the UNF queries.
These rules allow \sys to prove the equivalence of a larger set of \sql queries.

\PP{\dcircle{1} Empty Table:}
For an SPJ query, if the predicate is unsatisfiable, then this SPJ query becomes an empty table query.
Empty table query is a special query that always returns empty table.
This rule reduces a query to empty table query if all of its input sub-queries
are empty table queries.
\PP{\dcircle{2} Predicate Push-down:}
The filter predicate is always pushed to its input sub-queries, if it can be
pushed.

\PP{\dcircle{3} Aggregate Merge:}
If the input sub-queries of an aggregate query is also an aggregate query, then two aggregate query can be merged if 
the group set of outer aggregate query is a subset of group set of the inner aggregate query and 
the aggregate function is one of the following functions: 
\texttt{MAX}, \texttt{MIN}, \texttt{SUM} and \texttt{COUNT}.

\PP{\dcircle{4} Integrity Constraints:}
\sys supports integrity constraints by encoding them as normalization rules.  
If a table is joined with itself on its primary key, then \sys normalizes the
join operation to a projection. 
If the input query of an aggregate query is a single table, and 
the group by set only contains the primary key, and there is no aggregation
functions, then the aggregate query may be removed.

%% file: fig/Expr.tex
\begin{align*}
 \vspace{-0.2in}
 \expr ::=~&\TABLE{n}\;|\; \SPJ{\vec{\expr}}{\textsc{p}}{\vec{o}}\;|\;\SAggregate{\expr}{\vec{g}}{\vec{agg}}\;|\;\SUnion{\vec{\expr}}
 \vspace{-0.3in}
\end{align*}

%% file: fig/Expr2.tex
\begin{align*}
 \vspace{-0.1in}
 \UNF ::=~&\SUnion{\vec{\SPJE}} \; \SPJE ::=\SPJ{\vec{\exprE}}{\textsc{p}}{\vec{o}}\\
 \exprE ::= ~&\TABLE{n}\;|\;\SAggregate{\UNF}{\vec{g}}{\vec{agg}} 
\end{align*}

%% file: equivalence.tex
\section{Equivalence Verification}
\label{sec:equivalence}

In this section, we discuss how \sys verifies the equivalence of two
queries.
We first present the formal definitions of two types of equivalence under bag semantics
in~\autoref{sec:equivalence::definitions}.
We then describe how \sys proves the full equivalence of a pair of
cardinally equivalent queries using their QPSR in~\autoref{sec:equivalence::check}.
Lastly, we discuss how \sys verifies if a pair of queries are cardinally equivalent,
and constructs their QPSR when they are cardinally equivalent
in~\autoref{sec:equivalence::construct},
\autoref{sec:equivalence::construct:::table},
\autoref{sec:equivalence::construct:::spj},
\autoref{sec:equivalence::construct:::agg}, and
\autoref{sec:equivalence::construct:::union}.

\subsection {Equivalence Definitions}
\label{sec:equivalence::definitions}

To define the full equivalence relationship under bag semantics between two queries, 
we first define the cardinal equivalence relationship.
\begin{mydef}
  \label{def:cequivalent}
  \textsc{Cardinal Equivalence:}
  Given a pair of queries  \qa and \qb,
  \qa and \qb are \textit{cardinally equivalent} under bag semantics if and only if (iff), 
  for all valid input tables, the output tables $T_1$ and $T_2$ of \qa and \qb
  contain the same number of tuples.
\end{mydef}

If \qa and \qb are cardinally equivalent under bag semantics, for all valid inputs, each tuple in $T_1$
can be mapped to a unique tuple in $T_2$, and all tuples in $T_2$ are in the map.
Thus, it is a bijective (one-to-one) map between tuples in $T_1$ and $T_2$.
However, the two mapped tuples may differ in their values, as shown
in~\cref{fig:cardeq}.

\begin{mydef}
  \label{def:equivalent}
  \textsc{Full Equivalence:}
  Given a pair of queries \qa and \qb,
  \qa and \qb are \textit{fully equivalent} under bag semantics iff, 
  for all valid input tables $\inputs$, 
  the output tables $T_1$ and $T_2$ of \qa and \qb are identical.
\end{mydef}

Based on the definition, \qa and \qb are fully equivalent under bag semantics, 
iff there exists a bijective map between tuples in $T_1$ and $T_2$, and this bijective map is an identity map.
In other words, each tuple in $T_1$ can always be mapped to a unique, identical
tuple in $T_2$, and all tuples in $T_2$ are in the map, as shown
in~\cref{fig:fulleq}.
Thus, by proving the existence of a bijective, identity map between tuples in $T_1$ and $T_2$ for all valid inputs, we prove 
\qa and \qb are fully equivalent under bag semantics.

\PP{Motivation:}
\sys first \textit{quickly} checks for cardinal equivalence before checking for
full equivalence.
This is because if \qa and \qb are fully equivalent, then they must be
cardinally equivalent.
If \qa and \qb are cardinally equivalent, then there always exists a bijective map between tuples in output tables for all valid inputs.
To prove \qa and \qb are fully equivalent, we only need to prove that the bijective map between tuples in the
output tables is always an identity map for all valid inputs.
In the rest of the paper, \textit{equivalent} queries without any qualifier
refer to fully-equivalent queries.

\sys can prove two queries are fully equivalent even if their sub-queries are
\textit{only} cardinally equivalent.
For example, while \qa and \qb in \cref{ex:aggregate} are fully equivalent,
their sub-queries are not fully equivalent.
This is because the inner query in \qa and \qb returns two and three columns,
respectively.

\subsection {Query Pair Symbolic Representation}
\label{sec:equivalence::check}
We now define the QPSR of two cardinally equivalent queries that \sys uses for
proving QE.
QPSR is used to symbolically represent a bijective map between the tuples that are returned by two cardinally
equivalent queries for all valid inputs.
QPSR of a pair of cardinally equivalent queries \qa and \qb is a tuple
of the form:
\[ \langle \vec{\tabcol_1}, \vec{\tabcol_2}, \cond, \assign \rangle \]
% column repr:
\PPS{$\vec{\tabcol_1}$} is a vector of pairs of FOL terms that represent an
arbitrary tuple returned by \qa.
Each element of this vector represents a column and is of the form: $(\tabval,
\tabisnull)$, where \tabval represents the value of the column and  
\tabisnull denotes the nullability of the column.
% column repr:
\PPS{$\vec{\tabcol_2}$} is another vector of pairs of FOL terms that
represents the corresponding tuple returned by \qb.
Since \qa and \qb must be cardinally equivalent before \sys constructs
their QPSR, a bijective map exists between the returned tuples for all
valid inputs, which is symbolically represented by the two tuples
$\vec{\tabcol_1}$ and $\vec{\tabcol_2}$.
% constraint formula:
\PPS{\cond} is an FOL formula that represents the constraints 
that must be satisfied for the symbolic tuples $\vec{\tabcol_1}$ 
and $\vec{\tabcol_2}$ to be returned by \qa and \qb, respectively.
They encode the semantics of the predicates in the queries.
%  
% assignment formula:
\PPS{\assign} is another FOL formula that
specifies the relational constraints between symbolic variables
used in $\vec{\tabcol_1}$, $\vec{\tabcol_2}$ and \cond.
This formula is used for supporting complex SQL operators, such as
\texttt{CASE} expression.

\begin{mydef}
  \label{def:symbolic-bijective}
  \textsc{Symbolic Bijective:} Given a pair of cardinally equivalent
  queries \qa and \qb, the pair of symbolic tuples ($\vec{\tabcol_1}$,
  $\vec{\tabcol_2}$) in the QPSR symbolically represents a bijective
  map between tuples returned by \qa and \qb iff, for a valid input
  tables $\inputs$, the pair of tuple ($t_1$, $t_2$) in a
  bijective map, where $t_1$ is returned by evaluating \qa on $\inputs$, and $t_2$
  is returned by evaluating \qb on $\inputs$ , there exists a model $m$ such that
  $m[\vec{\tabcol_1}] = t_1$ and $m[\vec{\tabcol_2}] = t_2$.
\end{mydef}

Based on \cref{def:symbolic-bijective}, for a given set of input
tables, the two tuples $t_1$ and $t_2$, which are in a bijection
between two output tables of two cardinally equivalent queries, are
equivalent to the interpretation of two symbolic tuples
$\vec{\tabcol_1}$ and $\vec{\tabcol_2}$ under the same model.
A model is a set of concrete values for the symbolic variables in the two
symbolic tuples.
$m[\vec{\tabcol_1}]$ denotes the interpretation of tuple
$\vec{\tabcol_1}$ on model $m$, which substitutes the variables in
$\vec{\tabcol_1}$ by the corresponding values in $m$.  

\PP{Verifying Full Equivalence:}
To prove that two cardinally equivalent queries \qa and \qb are fully equivalent,
\sys needs to prove that the bijective map between returned tuples is always an identity map.
In other words, \sys needs to prove that, 
for an arbitrary tuple $t$ returned by \qa, the bijective map associates $t$ to an
identical tuple returned by \qb with the same values.
\sys verifies this property using the QPSR of \qa and \qb.
When both symbolic tuples satisfy the predicate (\ie, $\cond$),
it must verify that $\vec{\tabcol_1}$ is equivalent to $\vec{\tabcol_2}$.
This property is formalized as: 
\[\cond \land \assign \implies \vec{\tabcol_1} = \vec{\tabcol_2}  \]
\sys verifies this property using an SMT solver~\cite{moura08}.
If the property does not hold, then the negation of this property is satisfiable.
\sys feeds the negation of this property into the SMT solver.
If the solver determines that this formula is unsatisfiable, then it
determines that $\vec{\tabcol_1}$ and $\vec{\tabcol_2}$ are always
identical.
In this manner, \sys leverages the QPSR to prove full equivalence.

\begin{lemma}
\label{lemma:fulleqsound}
Given a QPSR of queries \qa and \qb, if
\[ \cond \land \assign \land \neg (\vec{\tabcol_1} =
  \vec{\tabcol_2}) \]
is unsatisfiable, then \qa and \qb are fully equivalent.
\end{lemma}

\begin{proof}
Since $\cond \land \assign \land \neg (\vec{\tabcol_1} = \vec{\tabcol_2})$ is
unsatisfiable, 
there is no model $m_1$ such that $m_1[\vec{\tabcol_1}] \neq
m_2[\vec{\tabcol_2}]$ and $m_1[\vec{\tabcol_1}] = t_1$ where $t_1$ is a tuple in
the output table of \qa for all valid inputs.  
Thus, for a model $m$, if $m[\vec{\tabcol_1}] = t_1$, then $m[\vec{\tabcol_1}]
=  m[\vec{\tabcol_2}]$. 
Based on the definition of QPSR, $\vec{\tabcol_1}$ and $\vec{\tabcol_2}$
symbolically represents a bijective map between tuples returned by \qa and \qb. 
Based on \cref{def:symbolic-bijective}, $t_1 = t_2$ for all valid inputs where 
tuple $t_1$ and $t_2$ in a bijective map between two output tables.
Based on the definition of full equivalence in \cref{def:equivalent}, \qa and
\qb are fully equivalent. 
\end{proof}
\subsection{Construction of QPSR}
\label{sec:equivalence::construct}

\begin{algorithm}[t]
  \footnotesize
  \SetKwInOut{Input}{Input}
  \SetKwInOut{Output}{Output}
  \SetKwProg{myproc}{Procedure}{}{}
  \Input{A pair of queries (\ie, \qa and \qb)}
  \Output{QPSR of \qa and \qb or \nul}
  \myproc{$\veriCard(\qa,\qb)$}{ 
       \Switch{\texttt{TypeOf}($\qa,\qb$)}{%
         \lCase{$Table$}{
           \Return{$\veriCardTable(\qa,\qb)$}
         }
         \lCase{$SPJ$}{
            \Return{$\veriCardSPJ(\qa,\qb)$}
         }
         \lCase{$Union$}{
            \Return{$\veriCardUnion(\qa,\qb)$}
         }
         \lCase{$Agg$}{
            \Return{$\veriCardAgg(\qa,\qb)$}
         }
         \lCase{\texttt{Type Mismatch}}{
           \Return{\nul} 
         }
       }
  }
  \caption{Procedure for verifying the cardinal equivalence of queries.
  It returns the QPSR if and only if the queries are cardinally equivalent.}
  \label{alg:check1}
\end{algorithm}

\autoref{alg:check1} presents a recursive procedure $\veriCard$ for verifying
the cardinal equivalence of two queries.
The $\veriCard$ procedure takes a pair of queries as inputs (\ie, $\qa$ and
$\qb$).
\sys first checks the types of the given queries.
If they are of the same type, then it invokes the appropriate sub-procedure for
that particular type.
We describe these four sub-procedures in \cref{sec:equivalence::construct:::table,sec:equivalence::construct:::spj,sec:equivalence::construct:::agg,sec:equivalence::construct:::union}.
If $\qa$ and $\qb$ are cardinally equivalent, then $\veriCard$ returns their
QPSR.
If these queries are of different types, it returns \nul to indicate that it cannot
determine their cardinal equivalence.
This is because each type of queries has different semantics
(\autoref{sec:approach::ar}).

Some sub-procedures recursively invoke $\veriCard$ to verify the
cardinal equivalence between their sub-queries.
\sys applies the normalization rules defined in~\autoref{sec:approach} to
transform the given two queries so that they are of the same type (and the 
sub-queries are also of the same types recursively).
This normalization process is \textit{incomplete} (\ie, \sys may conclude that
two queries are not cardinally equivalent since they cannot be normalized
to the same type, even if they are actually cardinally equivalent).
We discuss this limitation in~\autoref{sec:evaluation-limitations}.

Each sub-procedure takes a pair of queries of the same type 
as inputs.
It first attempts to determine if they are cardinally equivalent.
If they are cardinally equivalent, then it constructs the QPSR of
\qa and \qb.
Otherwise, it returns \nul to indicate that it cannot determine their cardinal
equivalence.
In each of the following sub-sections, we first describe the conditions that 
are sufficient for proving cardinal equivalence based on the semantics of the query.
We then describe how each sub-procedure verifies these conditions to prove cardinal equivalence.
We then discuss how \sys constructs the QPSR if they are cardinally
equivalent.
Lastly, we describe their soundness and completeness properties
\footnote{
A sub-procedure $P$ is sound if whenever it returns a QPSR, the given two queries are 
cardinally equivalent and the two symbolic tuples symbolically represents a bijective map for all valid inputs.
A sub-procedure $P$ is complete if whenever it returns \nul, the given two queries 
are not cardinally equivalent.
}.

\subsection{Table Query}
\label{sec:equivalence::construct:::table}

\autoref{alg:tableCheck} illustrates the $\veriCardTable$ procedure for table query.

\begin{lemma}
\label{lemma:tableAR}
A pair of table queries $\TABLE{n_1}$ and $\TABLE{n_2}$ are cardinally
equivalent iff their input tables are the same. (\ie, $n_1$ = $n_2$).
\end{lemma}
\PP{Cardinal Equivalence:}
Since the table query returns all tuples from the input table,
thus if two table queries have the same input table, then they will always 
have the same number of tuples.
So $\veriCardTable$ compares the names of the two input tables.
\PP{QPSR:}
We define the QPSR of the two cardinally equivalent table queries using an
\textit{identity map} between the returned tuples (\eg, QPSR-1 in
\cref{ex:table::qpsr}).
$\veriCardTable$ first constructs the symbolic tuple $\vec{\tabcol_1}$ using  
a vector of pairs of variables based on $n_1$'s table schema $\tabschema(n_1)$,
and then sets the symbolic tuple $\vec{\tabcol_2}$ to be the same as
$\vec{\tabcol_1}$.
These two equivalent tuples $\vec{\tabcol_1}$ and $\vec{\tabcol_2}$
define a bijective map between the returned tuples.
$\veriCardTable$ sets the $\cond$ and $\assign$ fields as \tru since there are
no additional constraints that the tuples in the table must satisfy.

\begin{algorithm}[t]
  \footnotesize
  \SetKwInOut{Input}{Input}
  \SetKwInOut{Output}{Output}
  \SetKwProg{myproc}{Procedure}{}{}
  \Input{A pair of table queries}
  \Output{QPSR of two queries or \nul}
  \myproc{$\veriCardTable(\TABLE{n_1},\TABLE{n_2})$}{
       \uIf{ $n_1 = n_2$ }{
           $\vec{\tabcol_1} \gets \fresh(\tabschema(n_1))$ \\
           $\vec{\tabcol_2} \gets \vec{\tabcol_1}$ \\
           \Return{($\vec{\tabcol_1}$,$\vec{\tabcol_2}$,$\tru$,$\tru$)}
       }
       \lElse {\Return{\nul}}
     }
     \caption{Verification Algorithm for Table queries}
  \label{alg:tableCheck} 
\end{algorithm}

\PP{Properties:}
$\veriCardTable$ is sound and complete.
These two properties directly follow from~\autoref{lemma:tableAR}.
We present a formal proof in the appendix~\autoref{app:table-bijs}. 

\subsection{SPJ Queries}
\label{sec:equivalence::construct:::spj}
\begin{algorithm}[t]
  \scriptsize
  \SetKwInOut{Input}{Input} 
  \SetKwInOut{Output}{Output}
  \SetKwProg{myproc}{Procedure}{}{}
  \Input{A pair of SPJ Queries}
  \Output{QPSR of given SPJ queries or \nul}
  \myproc{$\veriCardSPJ(\SPJ{\vec{e_1}}{p_1}{\vec{o_1}},\SPJ{\vec{e_2}}{p_2}{\vec{o_2}})$}{
       $\{\vec{QPSR}\} \gets \veriCardExpr(\vec{e_1},\vec{e_2})$ \label{ln:spj-joins} \\
       \ForEach{$\vec{QPSR} \in \{\vec{QPSR}\}$ \label{ln:spj-iter} }{
           $(\vec{\tabcol_1},\vec{\tabcol_2},\cond,\assign) \gets \Compose(\vec{QPSR})$ \label{ln:construct} \\
           $(\cond_1,\assign_1) \gets \constructPred(p_1,\vec{\tabcol_1})$ \label{ln:predicate1} \\
           $(\cond_2,\assign_2) \gets \constructPred(p_2,\vec{\tabcol_2})$ \label{ln:predicate2} \\
           \If{$\cond_1  \leftrightarrow \cond_2$}{ \label{ln:prove}
            $(\vec{\tabcol'_1},\assign_3) \gets \conexpr(\vec{\tabcol_1},\vec{o_1})$ \label{ln:spj-def-assn} \\
            $(\vec{\tabcol'_2},\assign_4) \gets \conexpr(\vec{\tabcol_2},\vec{o_2})$ \\
            $\cond \gets \cond_1 \land \cond_2 \land \cond$ \\
            $\assign \gets \assign \land \assign_1 \land \assign_2 \land \assign_3 \land \assign_4$ \\
            \Return{ ($\vec{\tabcol'_1},\vec{\tabcol'_2},\cond,\assign$) } \label{ln:spj-ret-qpsr}
           }
         }
       \Return{\nul} \label{ln:spj-ret-null}
  } \caption{Verification Algorithm for SPJ queries}
  \label{alg:SPJCheck}
\end{algorithm}

\autoref{alg:SPJCheck} illustrates the $\veriCardSPJ$ procedure for SPJ queries.
If this procedure determines that the two input SPJ queries
$\SPJ{\vec{e_1}}{p_1}{\vec{o_1}}$ and $\SPJ{\vec{e_2}}{p_2}{\vec{o_2}}$ are
cardinally equivalent, then it returns their QPSR.
Otherwise, it returns \nul.
$\veriCardSPJ$ leverages two procedures from~\cite{zhou19}: $\conexpr$ and
$\constructPred$ .

$\conexpr$ takes a vector of projection expressions and a symbolic tuple as  
inputs, and returns a new symbolic tuple with additional constraints $\assign$
that models the relation between variables.
This new symbolic tuple represents the modified tuple based on the vector of
projection expressions.
$\constructPred$ takes a predicate and a symbolic tuple as the input and returns
a boolean formula $\cond$ with additional constraints $\assign$.
$\cond$ symbolically represents the result of evaluating the predicate on the
symbolic tuples.
$\constructPred$ supports higher-order predicates, such as \texttt{EXISTS}, by
encoding them as an uninterpreted function.

\PP{Cardinal Equivalence:}
As covered in~\autoref{sec:approach::ar}, an SPJ query first computes the cartesian
product of all input sub-queries as the intermediate table (\qjoin).
It then selects all tuples in the intermediate table that satisfy the predicate
(\qscan), and applies the projection on each selected tuple (\qproj).

\begin{lemma}
\label{lemma:spjAR}
A pair of SPJ queries $\SPJ{\vec{e_1}}{p_1}{\vec{o_1}}$ and
$\SPJ{\vec{e_2}}{p_2}{\vec{o_2}}$ are cardinally equivalent 
if there is a bijective map $m$ between tuples in intermediate join tables, 
such that the predicates $p_1$ and $p_2$ always return the same
result for the corresponding tuples in $m$.

\end{lemma}
To prove that there is a bijective map between the tuples in the two
intermediate join tables, $\veriCardSPJ$ first uses the $\veriCardExpr$
procedure to find a bijective map between input sub-queries such that each pair of sub-queries
are cardinally equivalent.
$\veriCardExpr$ exhaustively examines all possible maps and recursively uses
$\veriCard$ to verify the cardinal equivalence between two sub-queries.
$\veriCardExpr$ returns all possible candidate maps wherein each pair of sub-queries
are cardinally equivalent ($\{\vec{QPSR}\}$).
Each candidate map is represented by a vector of QPSR ($\vec{QPSR}$),
wherein each QPSR defines a bijective map between tuples returned by a pair of
cardinally equivalent sub-queries.

$\veriCardSPJ$ then uses the $\Compose$ procedure to construct two symbolic
tuples $\vec{\tabcol_1}$ and $\vec{\tabcol_2}$ (\autoref{ln:construct}) that
represent a bijective map between the tuples in the two intermediate join
tables.
These two symbolic tuples are constructed by concatenating symbolic tuples from
the QPSRs of sub-queries based on the order of sub-queries in the input vectors.
$\Compose$ also constructs $\cond$ and $\assign$ by taking the conjunction of 
$\cond$ and $\assign$ from the QPSRs of sub-queries, respectively.

$\veriCardSPJ$ then tries to prove that the two predicates always return the same result
for the two symbolic tuples.
$\veriCardSPJ$ first leverages the $\constructPred$ procedure to encode 
predicates $p_1$ and $p_2$ on $\vec{\tabcol_1}$ and $\vec{\tabcol_2}$,
respectively (\autoref{ln:predicate2}).
$\veriCardSPJ$ uses an SMT solver to prove this property under sub-conditions $\cond$ and all 
relational constraints: $\assign$, $\assign_1$, $\assign_2$
(\autoref{ln:prove}).
If the property holds, then negation of this property is unsatisfiable:
\[\cond \land \assign \land \assign_1 \land \assign_2 \land \neg (\cond_1 = \cond_2) \]
$\veriCardSPJ$ feeds this formula to an SMT solver.
If the solver determines that this formula is unsatisfiable, then we
prove $\cond_1$ and $\cond_2$ are always equivalent when the
relational constraints $\assign_0$, $\assign_1$, and $\assign_2$ and sub-conditions $\cond$ hold.

\begin{figure}[t!]
\centering
\includegraphics[width=0.95\linewidth]{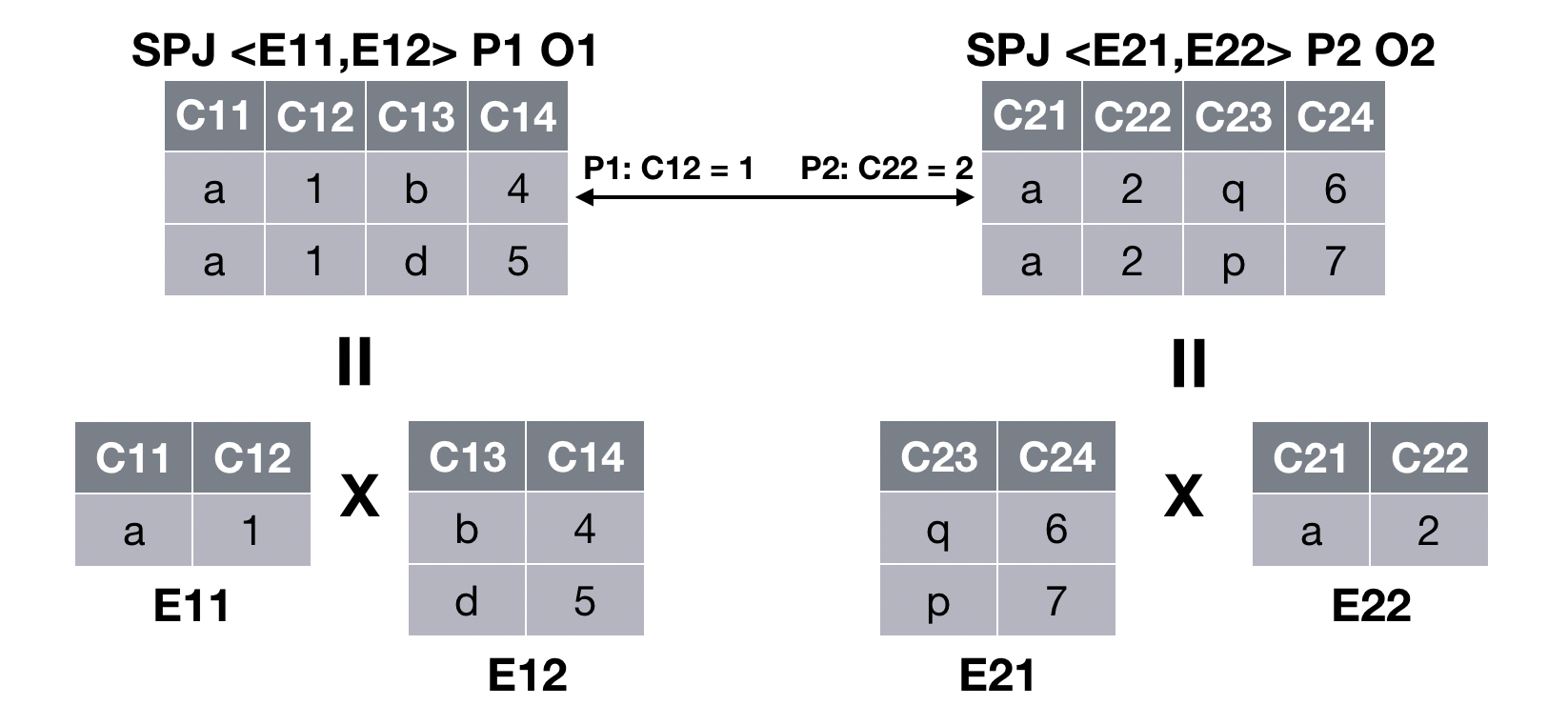}
\caption{
 \textbf{SPJ Queries} -- Cardinally equivalent SPJ Queries.
}
\label{fig:spjeq}
\end{figure}

Consider the cardinally equivalent SPJ queries shown in~\cref{fig:spjeq}.
In this case, $\veriCardSPJ$ first verifies that sub-query $E11$ is cardinally
equivalent to sub-query $E22$, and sub-query $E12$ is cardinally equivalent to sub-query
$E21$.
Thus, the two intermediate join tables (\ie, cartesian product of sub-tables)
are cardinally equivalent.
$\veriCardSPJ$ constructs two symbolic tuples to represent the bijective map
between these intermediate join tables by leveraging the two bijective maps
between the underlying tables. 
$\veriCardSPJ$ then verifies that two corresponding tuples in the map either
both satisfy the predicate or do not satisfy the predicate. 
Thus, the bijective map between the tuples in the intermediate join tables is
the bijective map between the tuples in the output tables before projection.

\PP{QPSR:}
Since $\veriCardSPJ$ verifies that the given pair of SPJ queries are cardinally equivalent, 
the two symbolic tuples $\vec{\tabcol_1}$ and $\vec{\tabcol_2}$ define a
bijective map between tuples in the output tables before projection.
Projection does not change the bijective map between tuples as it is applied
separately on each tuple.
Thus, $\veriCardSPJ$ leverages $\conexpr$  
to construct new symbolic tuples $\vec{\tabcol'_1}$ and $\vec{\tabcol'_2}$ 
based on the vector of projection expressions and the given symbolic tuples.
The QPSR consists of the derived symbolic tuples $\vec{\tabcol'_1}$,
$\vec{\tabcol'_2}$, the conjunction of $\cond_1$, $\cond_2$ and $\cond$, and 
the conjunction of all the relational constraints.

\PP{Properties:}
$\veriCardSPJ$ is sound.
Based on~\autoref{lemma:spjAR}, if $\veriCardSPJ$ returns the QPSR, then the 
given SPJ queries are cardinally equivalent.
We present a formal proof in~\autoref{app:spj-bijs}.

In general, $\veriCardSPJ$ is \textit{not complete}.
The reasons are threefold.
First, the SMT solver is only complete for linear operators.
If the predicates have non-linear operators (\eg, multiplication between
columns), then the solver may return \texttt{UNKNOWN} when it should return
\texttt{UNSAT}~\cite{zhou19}.
Second, \sys encodes all user-defined functions, string operations, and
higher-order predicates as uninterpreted functions.
These encodings do not preserve the semantics of these operations.
Third, $\veriCard$ is not complete (\autoref{sec:equivalence::construct}).

$\veriCardSPJ$ procedure is complete if all input queries for the given two SPJ queries
are table queries, and the SMT solver can determine the satisfiability of the
predicates.
This is because the problem of deciding equivalence of two conjunctive (\ie,
SPJ) queries is decidable~\cite{cohen98}.
\begin{lemma}
\label{lemma:spj-complete}
For a given pair of cardinally equivalent SPJ queries $\SPJ{\vec{e_1}}{p_1}{\vec{o_1}}$ and
$\SPJ{\vec{e_2}}{p_2}{\vec{o_2}}$, if $\vec{e_1}$ and $\vec{e_1}$ only have table sub-queries,
and the SMT solver can determine the satisfiability of the predicates and projection expressions, then $\veriCardSPJ$ procedure returns the QPSR.
\end{lemma}

\begin{proof}
We prove this theorem using the method of contraposition.
If $\veriCardSPJ$ returns \nul, then there are two cases:

\textbf{Case 1:} 
There is no bijective map between $\overrightarrow{e_1}$ and $\overrightarrow{e_2}$, 
such that each pair of table sub-queries are cardinally equivalent.
There are two possible sub-cases.

In the first sub-case, $\overrightarrow{e_1}$ has more input table queries than
$\overrightarrow{e_2}$. 
Here, we can always construct the input such that the
intermediate join table of $\vec{e_1}$ has more tuples than the
intermediate join table of $\vec{e_2}$.
\sys has eliminated the case where the predicates are $False$ 
(\autoref{sec:approach::unf}).
Thus, the constructed tuples in intermediate join table can all satisfy
the predicate.
So, the number of tuples returned by the two SPJ queries is different, 
and are hence not cardinally equivalent.
In the second sub-case, $\overrightarrow{e_1}$ have different table queries
than $\overrightarrow{e_2}$.
Here, each input table in $\overrightarrow{e_2}$ contains only one tuple. 
And a different table in $\overrightarrow{e_1}$ has two tuples that satisfy
the predicate.
So, the number of tuples returned by the two SPJ queries is different, 
and are hence not cardinally equivalent.

\textbf{Case 2:} 
$\veriCardSPJ$ cannot verify that the two predicates always return the same
result for two corresponding tuples in a bijective map between tuples in the intermediate table.
In this case, since both predicates are decidable, the solver will generate a
model $m$ such that one symbolic tuple satisfies the predicate and the other
one does not.
We then construct inputs such that each intermediate table only contains one
tuple that matches the values in $m$.
Then the first SPJ query returns a table that contains one tuple,
while the other one returns an empty table.
Thus, the two SPJ queries are not cardinally equivalent.
\end{proof}

\subsection{Aggregate Queries}
\label{sec:equivalence::construct:::agg}
\begin{algorithm}[t]
  \footnotesize
  \SetKwInOut{Input}{Input}
  \SetKwInOut{Output}{Output}
  \SetKwProg{myproc}{Procedure}{}{}
  \Input{A pair of aggregate queries}  
  \Output{QPSR of given aggregate queries or \nul}
  \myproc{$\veriCardAgg(\SAggregate{e_1}{\vec{g_1}}{\vec{agg_1}},
  \SAggregate{e_2}{\vec{g_2}}{\vec{agg_2}})$}{ $QPSR \gets \veriCard(e_1,e_2)$ \label{ln:agginputEq}\\
       \If{QPSR != \nul \label{ln:agg-chk} }{
           $(\vec{\tabcol_1},\vec{\tabcol_2},\cond,\assign) \gets QPSR $ \\
           \If{$\vec{g_1} \leftrightarrow \vec{g_2}$ \label{ln:chk-gps} }{
             \label{ln:proveGroup}
             $\vec{\tabcol_1} \gets \freshAgg(\vec{agg_1})$ :: $\vec{g_1}$ \label{ln:aggNewSymbolicTuples1}\\
             $\vec{\tabcol_2} \gets
             \constructAggCall(\vec{agg_1},\vec{\tabcol_1},\vec{agg_2})$ :: $\vec{g_2}$ \label{ln:aggNewSymbolicTuples2}\\
             \Return{($\vec{\tabcol_1},\vec{\tabcol_2},\cond,\assign$)} \label{ln:agg-ret-qpsr}
           }
         }
       \lElse {\Return{\nul}}
  } \caption{Verification Algorithm for Aggregate Queries}
  \label{alg:AggCheck}
\end{algorithm}
\autoref{alg:AggCheck} illustrates the $\veriCardAgg$ procedure for aggregate
queries.
If this procedure determines that the two input aggregate queries
$\SAggregate{e_1}{\vec{g_1}}{\vec{agg_1}}$ and $\SAggregate{e_2}{\vec{g_2}}{\vec{agg_2}}$ are
cardinally equivalent, then it returns their QPSR.
Otherwise, it returns \nul.

\PP{Cardinal Equivalence:}
An aggregate query groups the tuples in the input table based on the
\texttt{GROUP BY} column set, and then returns a tuple by applying the aggregate
function on each group. 
\begin{lemma}
\label{lemma:aggAR}
Two aggregate queries
$\SAggregate{e_1}{\vec{g_1}}{\vec{agg_1}}$ and 
$\SAggregate{e_2}{\vec{g_2}}{\vec{agg_2}}$ 
are cardinally equivalent if two conditions are satisfied:
\textbf{(1)} the two input sub-queries $e_1$ and $e_2$ are cardinally
equivalent;
\textbf{(2)} for any two pairs of corresponding tuples in a bijective map of the QPSR of $e_1$ and $e_2$,
two tuples in $e_1$ belong to the same group as defined by $g_1$ iff their associated tuples in $e_2$ belong to the same group
as defined by $g_2$.
\end{lemma}
$\veriCardAgg$ first recursively invokes the $\veriCard$ procedure to determine
the cardinal equivalence of the two input sub-queries $e_1$ and $e_2$
(\autoref{ln:agginputEq}).
If $\veriCard$ returns the QPSR of $e_1$ and $e_2$, then
$\veriCardAgg$ has proved the first condition in~\autoref{lemma:aggAR}.

To prove the second condition, $\veriCardAgg$ collects the symbolic tuples
$\vec{\tabcol_1}$ and $\vec{\tabcol_2}$ from the QPSR.
Since these two symbolic tuples represent a
bijective map between tuples returned by $e_1$ and $e_2$,
$\veriCardAgg$ replaces all variables in $\vec{\tabcol_1}$ and $\vec{\tabcol_2}$
by a set of fresh variables to generate a second pair of symbolic tuples 
$\vec{\tabcol'_1}$ and $\vec{\tabcol'_2}$ that represents the same bijective map with different tuples.

We decompose the proof for the second condition into two stages
(\autoref{ln:proveGroup}).
In the first stage, we want to prove that if $\vec{\tabcol_1}$ and
$\vec{\tabcol'_1}$ belong to the same group, then $\vec{\tabcol_2}$ and 
$\vec{\tabcol'_2}$ also belong to the same group.
To prove this, $\veriCardAgg$ extracts the \texttt{GROUP BY} column sets 
$\vec{g_1}$, $\vec{g'_1}$, $\vec{g_2}$ and $\vec{g'_2}$ from 
$\vec{\tabcol_1}$, $\vec{\tabcol'_1}$, $\vec{\tabcol_2}$ and $\vec{\tabcol'_2}$,
respectively.
It then attempts to prove the property: 
\[(\cond \land \assign \land \vec{g_1} = \vec{g'_1}) \implies  \vec{g_2} = \vec{g'_2}  \]
$\veriCardAgg$ sends the negation of this property to the solver.
If the solver decides that this formula is unsatisfiable, 
then it is impossible to find two tuples returned by $e_1$ that are assigned to
the same group by $\vec{g_1}$, such that their corresponding tuples returned by
$e_2$ are assigned to different groups by $\vec{g_2}$.
In the second stage, we use the same technique in the reverse direction of the
implication.

\begin{figure}[t!]
\centering
\includegraphics[width=0.85\linewidth]{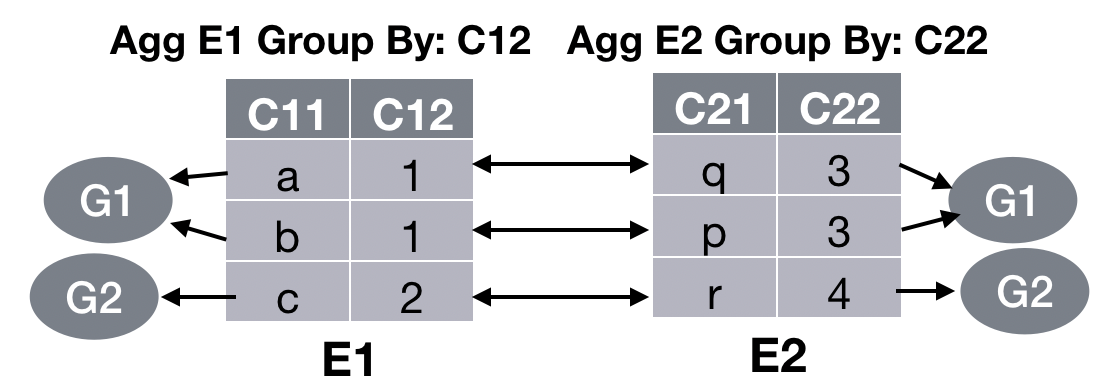}
\caption{
 \textbf{Aggregate Queries} -- Cardinally equivalent aggregate queries.
}
\label{fig:aggeq}
\end{figure}

Consider the cardinally equivalent aggregate queries shown in~\cref{fig:aggeq}.
$\veriCardAgg$ first verifies that the two input queries $E1$ and $E2$ are
cardinally equivalent, and then constructs the QPSR to represent the bijective
map between their returned tuples.
$\veriCardAgg$ then verifies that if two arbitrary tuples in $E1$ belong to same
group (\eg, first two tuples), then the two corresponding tuples in $E2$ also
belong to the same group.
It also verifies that if two arbitrary tuples in $E1$ belong to different
groups (\eg, first and third tuples), then the two corresponding 
tuples in $E2$ also belong to different groups.
$\veriCardAgg$ verifies two aggregate queries are cardinally equivalent by
verifying that they emit the same number of groups.

\PP{QPSR:}
$\veriCardAgg$ constructs the QPSR of two given aggregate queries after
proving they are cardinally equivalent.
$\vec{\tabcol_1}$ and $\vec{\tabcol_2}$ define a bijective map between 
tuples returned by input queries, and can also be used to define a bijective map
between groups in two aggregate queries. 
If two aggregate functions in $\vec{agg_1}$ and $\vec{agg_2}$ are the same and
operate on same values (\ie, input columns of the symbolic tuples are the same), 
then the aggregate values in the output tuples are the same, since each group
contains the same number of tuples.  

$\veriCardAgg$ invokes the $\freshAgg$ procedure on $\vec{agg_1}$ to construct a
vector of pairs of new symbolic variables as the symbolic tuples for aggregate functions.
In each pair of symbolic variables, the first variable represents the aggregate
value.
The second variable indicates if the aggregate value is \nul.
$\veriCardAgg$ concatenates the \texttt{GROUP BY} column set $\vec{g_1}$ with
the symbolic tuple $\vec{\tabcol_1}$.
$\veriCardAgg$ then invokes the $\constructAggCall$ procedure to construct
the symbolic columns for $\vec{agg_2}$,
and then concatenates with the \texttt{GROUP BY} column set $\vec{g_2}$.
$\constructAggCall$ uses the same pairs of symbolic variables for all 
aggregation operations in $\vec{agg_2}$, where the aggregation function type and
operand columns are the same in $\vec{agg_1}$.
$\veriCardAgg$ propagates $\cond$ and $\assign$ into the sub-QPSRs. 

\PP{Properties:}
$\veriCardAgg$ is sound.
Based on~\autoref{lemma:aggAR}, if $\veriCardAgg$ returns the QPSR, 
then the two given aggregate queries are cardinally equivalent.
This is because the two symbolic tuples $\vec{\tabcol_1}$ and
$\vec{\tabcol_2}$ are constructed from corresponding groups.
Thus, $\vec{\tabcol_1}$ and $\vec{\tabcol_2}$ define a bijective map between
tuples returned by the two aggregate queries.
We present a formal proof in~\cref{app:agg-bijs}.

$\veriCardAgg$ is not complete.
The sources of incompleteness are threefold:
(1) incompleteness of $\veriCard$,
(2) limitations of the SMT solver, and
(3) when $\veriCard$ returns the QPSR of two input sub-queries, 
the symbolic tuples in the QPSR define only one possible bijective map 
between tuples in the input tables. If $\veriCardAgg$ fails to prove the second condition
in~\autoref{lemma:aggAR}, it is still possible that there exists 
another bijective map that satisfies the second condition.

\subsection{Union Queries}
\label{sec:equivalence::construct:::union}
\begin{algorithm}[t]
  \footnotesize
  \SetKwInOut{Input}{Input}
  \SetKwInOut{Output}{Output}
  \SetKwProg{myproc}{Procedure}{}{}
  \Input{A pair of union queries} 
  \Output{QPSR of given two union queries or \nul}
  \myproc{$\veriCardUnion(\SUnion{\vec{e_1}},\SUnion{\vec{e_2}})$}{
       $\{\vec{QPSR}\} \gets \veriCardExpr(\vec{e_1},\vec{e_2})$ \label{ln:findPairUnion}\\
       \If{$\{\vec{QPSR}\}$ != $\emptyset$ }{ \label{ln:findUnionProof}
           $\vec{\tabcol_1} \gets \fresh()$; $\vec{\tabcol_2} \gets \fresh()$ \\
           $\vec{QPSR} \gets \{\vec{QPSR}\} $ \\
           $(\cond, \assign) \gets \PN(\vec{QPSR},\vec{\tabcol_1},\vec{\tabcol_2}) $ \\
           \Return{($\vec{\tabcol_1},\vec{\tabcol_2},\cond,\assign$)}
         }
       \lElse {\Return{\nul}}
  } \caption{Verification Algorithm for Union queries}
  \label{alg:UnionCheck}
\end{algorithm}
\autoref{alg:UnionCheck} illustrates the $\veriCardUnion$ procedure for
union queries.
If it determines that the union expressions are cardinally equivalent,
then it returns their QPSR. 
Otherwise, it returns \nul .

\begin{lemma}
\label{lemma:unionAR}
Two union queries $\SUnion{\vec{e_1}}$ and $\SUnion{\vec{e_2}}$
are cardinally equivalent if there exists a bijective map between the two input
sub-queries $\vec{e_1}$ and $\vec{e_2}$, such that each pair of queries are cardinally
equivalent.
\end{lemma}
\PP{Cardinal Equivalence:}
The lemma directly follows from the semantics of the union queries.
$\veriCardUnion$ procedure invokes $\veriCardExpr$
(\autoref{sec:equivalence::construct:::spj}) to find a bijective map between
$\vec{e_1}$ and $\vec{e_2}$ (\autoref{ln:findPairUnion}), such that each pair of queries are cardinally
equivalent.

\PP{QPSR:}
$\veriCardExpr$ finds all candidate bijective maps ($\{\vec{QPSR}\}$) 
between two input sub-queries $\vec{e_1}$ and $\vec{e_2}$, such that each pair of
sub-queries are cardinally equivalent.
In each candidate bijective map ($\vec{QPSR}$), a vector of QPSRs is constructed
such that each QPSR defines a bijective map between tuples returned by a pair of
sub-queries. 
$\veriCardUnion$ gets an arbitrary $\vec{QPSR}$ (\ie, one candidate bijective
map between the sub-queries).
It seeks to construct a bijective map between tuples returned by two union queries
that preserves all of the bijective maps between tuples returned by sub-queries in
that $\vec{QPSR}$.
It first constructs two fresh symbolic tuples 
$\vec{\tabcol_1}$ and $\vec{\tabcol_2}$.
It then invokes the $\PN$ procedure to set $\assign$ such that both
$\vec{\tabcol_1}$ and $\vec{\tabcol_2}$ are always equivalent to the symbolic
tuples in one sub-QPSR returned by $\veriCardExpr$, and $\cond$ such that $\cond$ in one sub-QPSR holds when symbolic tuples
equal to the tuples in that sub-QPSR.
$\PN$ creates a vector of boolean variables to set these constraints. 
$\veriCardUnion$ returns these two symbolic tuples, $\cond$, and 
$\assign$ as the QPSR of the given union queries.
\PP{Properties:}
$\veriCardUnion$ is sound.
Based on~\autoref{lemma:unionAR}, if $\veriCardUnion$ returns the QPSR, 
then the two union queries are cardinally equivalent.
The symbolic tuples $\vec{\tabcol_1}$ and $\vec{\tabcol_2}$ define a bijective
map between tuples returned by two union queries that preserves all of the bijective
maps between tuples in their cardinally equivalent sub-queries.
The formal proof is given in \autoref{app:union-bijs}.

$\veriCardUnion$ is incomplete.
The sources of incompleteness are threefold:
(1) incompleteness of $\veriCard$,
(2) limitations of the SMT solver, and
(3) two union queries may be cardinally equivalent even if there is no
bijective map between their sub-queries such that each pair of sub-queries is
cardinally equivalent.

%% file: sound.tex
\section{Soundness and Completeness}
\label{sec:proof}

We now discuss the soundness and completeness of \sys for verifying the equivalence of  
two queries.

\PP{Soundness:}
\sys is sound.

\begin{lemma}
\label{lemma:sound1}
Given a pair of queries \qa and \qb, if $\veriCard$ returns a QPSR of
\qa and \qb, then \qa and \qb are cardinally equivalent and the pair of symbolic tuples $(\vec{\tabcol_1}, \vec{\tabcol_2})$ symbolically represents
a bijective map between tuples returned by \qa and \qb.
\end{lemma}
\begin{proof}
  Proof is by induction on \qa (the choice of \qa versus \qb is arbitrary).
  Each case on the form \qa is proved by the soundness of four
  sub-procedure and the inductive hypothesis in cases when \qa is
  composite of sub-queries.
\end{proof}

\begin{lemma}
\label{lemma:sound2}
Given two queries \qa and \qb, if \sys constructs the QPSR of two normalized queries, and checks the formula holds for the QPSR:
$\cond \land \assign \implies \vec{\tabcol_1} = \vec{\tabcol_2} $
then \qa and \qb are fully equivalent.
\end{lemma}

\begin{proof}
Based on \autoref{lemma:sound1}, if \sys constructs the QPSR, then $\vec{\tabcol_1}$ and 
$\vec{\tabcol_2}$ symbolically represents
a bijective map between tuples returned by cardinally equivalent \qa and \qb.
Thus, \qa and \qb are fully equivalent based on \autoref{lemma:fulleqsound}.
\end{proof}
\PP{Completeness:}
In general, \sys is not complete.
Sources of incompleteness are discussed in
\cref{sec:equivalence::construct,sec:equivalence::construct:::table,sec:equivalence::construct:::spj,sec:equivalence::construct:::agg,sec:equivalence::construct:::union},
and discuss their practical impact in
\autoref{sec:evaluation-limitations}.
However, \sys is complete for a pair of SPJ queries \qa and \qb that do not
have predicates or projection expressions whose satisfiability cannot be
determined by the SMT solver and the input sub-queries are only table queries.

\begin{proof}
Based on \autoref{lemma:spj-complete}, if \qa and \qb are fully equivalent, then
$\veriCard$ returns an QPSR. 
Based on \autoref{lemma:sound1}, $\vec{\tabcol_1}$ and $\vec{\tabcol_2}$
symbolically represent the bijective map between tuples returned by \qa and \qb
for all valid inputs.   
Since \qa and \qb do not have predicates and projection expressions whose 
satisfiability cannot be determined by the SMT solver,
If the SMT solver determines
$\cond \land \assign \land \neg(\vec{\tabcol_1} = \vec{\tabcol_2})$ is
satisfiable, then it generates a model $m$.
The model $m$ defines input tables from which \qa and \qb each return
non-identical tuples equal to the interpretation of their respective
symbolic tuples under $m$.
Thus, \qa and \qb are not fully equivalent.
By contradiction, \sys is complete.
\end{proof}

%% file: evaluation.tex
\section{Evaluation}
\label{sec:evaluation}

In this section, we describe our implementation and evaluation of \sys.
We begin with a description of our implementation
in~\autoref{sec:evaluation-implementation}.
We next report the results of a comparative analysis of \sys against
\eqs~\cite{zhou19} and \udp~\cite{chu2018axiomatic}, 
state-of-the-art automated QE verifiers based on SR and AR, respectively.
We then quantify the efficacy of \sys in identifying overlapping queries across
production SQL queries in~\autoref{sec:evaluation-summary}.
We conclude with the limitations of the current implementation of \sys 
in~\autoref{sec:evaluation-limitations}.

\begin{table*}[t!] 
\centering
\footnotesize
\input{fig/table.tex}
\caption{
 \textbf{Comparative analysis between \sys, \eqs, and \udp} - 
 The results include the number of query pairs in the \calcite benchmark that
 these tools support, the number of pairs whose equivalence they can prove, 
 and the average time they take to determine QE.
}
\label{tab:calcite}
\end{table*}

\subsection{Implementation}
\label{sec:evaluation-implementation}

\begin{figure}
\centering
\includegraphics[width=0.9\linewidth]{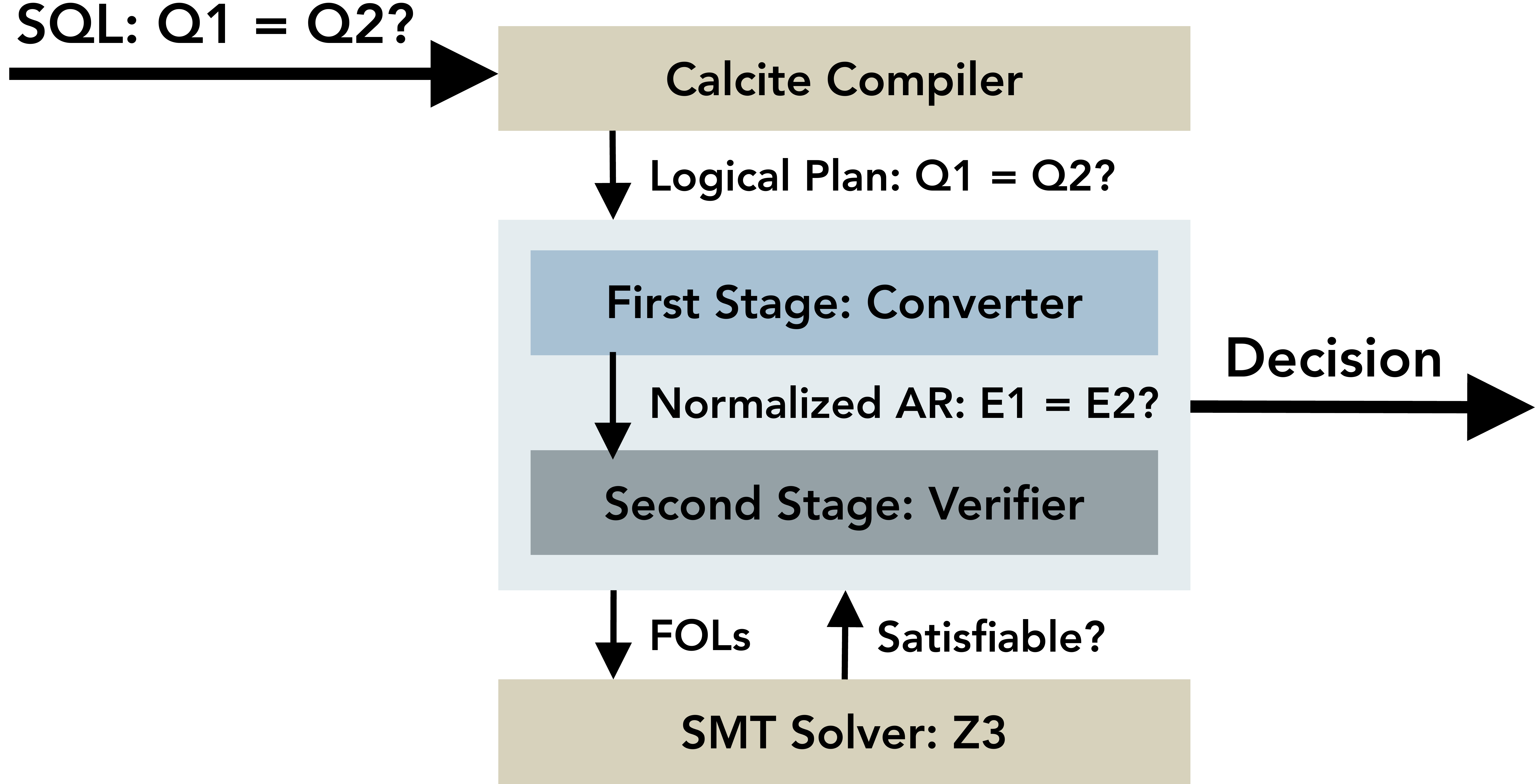}
\caption{ 
\textbf{Query Equivalence Verification Pipeline} -
The pipeline for determining the equivalence of SQL queries. 
}
\label{fig:structure}
\end{figure}

The architecture of \sys is illustrated in~\cref{fig:structure}.
\sys takes a pair of SQL queries as inputs and returns a boolean
decision that indicates whether they are fully equivalent.
The QE verification pipeline consists of three components:
\dcircle{1} The compiler converts the given queries to logical query 
execution plans.
We use the open-sourced \calcite framework~\cite{calcite}.
\dcircle{2} \sys operates on these logical plans in two stages.
First, it converts them to the categories of queries that it supports and
normalizes them.
Next, it uses the third component (\ie, an SMT solver) to verify the cardinal
equivalence of two queries and then constructs their QPSR.
It also uses the solver for verifying the properties of QPSR to determine full
equivalence.
We implement this component in Java (2,065 lines of code).
\dcircle{3} The third component is an SMT solver Z3 that \sys leverages for 
determining the satisfiability of FOL formulae~\cite{z3}.
We have published the source code of \sys on GitHub~\cite{spes}.

\subsection{Comparative Analysis}
\label{sec:evaluation-comparison}

\PP{Comparison Tools:}
We compare \sys against two QE verifiers: 
\eqs~\cite{zhou19} and \udp~\cite{chu2018axiomatic}.
\eqs takes a symbolic approach and only supports set semantics.
\udp takes an algebraic approach and does support bag semantics.

\PP{Benchmark:}
We use queries in the test suite of Apache \calcite~\cite{calcite} 
as our benchmark.
This test suite contains 232 semantically equivalent query pairs.
The reasons for using this benchmark are twofold. 
First, the \calcite optimizer is widely used in data processing
engines~\cite{drill,flink,hive,kylin,phoenix}.
So, it covers a wide range of SQL features\footnote{The test cases were obtained from the open-sourced \cos repository~\cite{cosette}.}.
Second, since \eqs and \udp are both evaluated on this query pair 
benchmark~\cite{chu2018axiomatic,zhou19}, 
we can quantitatively and qualitatively compare the efficacy of these tools.
We send every query pair with the schemata of their input tables 
to \sys and ask it to check their QE.
We conduct this experiment on a commodity server 
(Intel Core i7-860 processor and 16~GB RAM).

\PP{Automated SQL QE Verifiers:}
The results of this experiment are shown in~\cref{tab:calcite}.
We compare \sys against \eqs in the same environment.
We present the results reported in the \udp 
paper~\cite{chu2018axiomatic}\footnote{We were unable to conduct a comparative
performance analysis under the same environment since \udp is currently not
open-sourced.}.

We evaluate \sys under two configurations to delineate the impact of the
normalization rules.
\ding{182} \sys (w/o normalization) naively converts the original queries to a
tree representation without normalizing them, and then applies the verification
algorithms.
\ding{183} \sys additionally leverages a minimal set of normalization rules.

\sys proves the equivalence of a larger set of query pairs (95/232)
compared to \udp (34/232) and \eqs (67/232).
\sys currently supports 120 out of 232 pairs.
The un-supported queries either: 
(1) contain \sql features that are not yet supported (\eg, \cast), or 
(2) cannot be compiled by \calcite due to syntax errors.
Among the 120 pairs supported by \sys, 
it proves that 95 pairs (80\%) are equivalent under bag semantics.
In contrast, \udp proves the equivalence of 34 pairs under bag
semantics.
\eqs proves the equivalence of 67 pairs, but only under set semantics.
We group the proved query pairs into three categories:
\squishitemize
\item \textbf{USPJ:} Queries that are union of \qscan-\qproj-\qjoin.
\item \textbf{Aggregate:} Queries containing at least one aggregate.
\item \textbf{Outer-Join:} Queries containing at least one outer \qjoin.
\squishend

\cref{tab:calcite} reports the number of pairs proved by \udp and \eqs in each
category.
The number of proved pairs containing outer \qjoin is not known in case of \udp.
\sys outperforms the other tools on queries containing aggregate and outer
\qjoin operators.
\sys proves QE of a larger set of query pairs (95/232) compared to \sys (w/o
normalization) (56/232).
This illustrates the importance of the normalization stage of \sys for
real-world applications.

\PP{Efficiency:}
We next compare the average time taken by \sys and \eqs to prove the
equivalence of a pair of queries in each category.
This is an important metric for a cloud-scale tool that must be deployed in
a DBaaS platform.
We only compute this metric for the pairs that these tools can prove.
\sys and \eqs take 0.05~s and 0.15~s on average to prove QE, respectively.
So, \sys is 3$\times$ faster than \eqs on this benchmark.
\sys consistently outperforms \eqs across all categories of queries.

% \PP{Leveraging Materialized Views:}
% %
% In this experiment, we compare \sys against equivalence verification algorithms
% used in systems for leveraging materialized views:
% (1) \minicon~\cite{rachel00} and (2) \viewmatcher~\cite{goldstein01}.
% %

% \minicon only proves containment relationships between conjunctive queries 
% (\ie, SPJ queries).
% %
% So, it supports $30$ pairs of SPJ queries in the \calcite benchmark.
% %
% In contrast, \sys proves that 27 of these 30 query pairs are equivalent.
% %
% It does not support $3$ pairs since their equivalence is conditioned on 
% integrity constraints that \sys currently does not support.
% %
% \sys supports other types of queries in the \calcite benchmark that \minicon
% cannot support.

% \viewmatcher only proves containment relationships between SPJ queries and
% aggregate queries whose inputs are SPJ queries.
% %
% It leverages a syntactical comparison scheme to verify the containment
% relationship between queries.
% %
% We implemented this comparison scheme and found that it proves $25$ pairs of
% queries are equivalent in \calcite.
% %
% Since \sys relies on semantic comparison which subsumes syntactical
% comparison, it supports all of them.

%
\begin{table*}[t!]
\centering
\footnotesize
\input{fig/table2.tex}
\caption{
 \textbf{Efficacy of \sys on Production Queries} -
 "Highest Query Frequency" indicates the highest frequency of a query in equivalent 
 query pairs. 
 ``Compared Query Pairs" refers to number of query pairs that operate on the
 same set of input tables.
}
\label{fig:ant}
\end{table*}

\subsection{Efficacy on Production Queries}
\label{sec:evaluation-summary}

In this experiment, we quantify the efficacy of \sys in detecting 
overlap in production \sql queries.
We leverage three sets of real production queries from Ant
Financial~\cite{antfin}, a financial technology company.
These queries are used to detect fraud in business transactions.
In each set, we run \sys on each pair of queries that operate on 
the same set of input tables.
If \sys decides that a given pair of queries are not equivalent, 
then we check any constituent sub-queries that operate on the same input
tables.
We skip checking queries containing only table scans and those that 
only differ in the parameters passed on to their predicates.
This is because \sys trivially proves their equivalence and the computational
resources needed for evaluating such queries are negligible.

\cref{fig:ant} presents the results of this experiment.
\sys effectively identifies overlap between complex analytical queries.
Among 9486 queries, \sys finds overlapping computation between 2591 (27\%)
queries, while \eqs only finds overlapping computation between 1126 (12\%)
queries.
These tools report overlapping computation if the query or a constituent
sub-query is equivalent to another query or sub-query within the set.

We also report the highest frequency of queries present in these pairs 
that are repeatedly executed in the workload.
In practice, most of the computational resources are expended on executing
queries containing aggregate functions or different types of join.
Among 12090 equivalent pairs, 5831 (48\%) contain join and aggregate
operations.
This illustrates that \sys works well on queries containing these operators.
\begin{figure}[t!]
\centering
\includegraphics[width=0.95\linewidth]{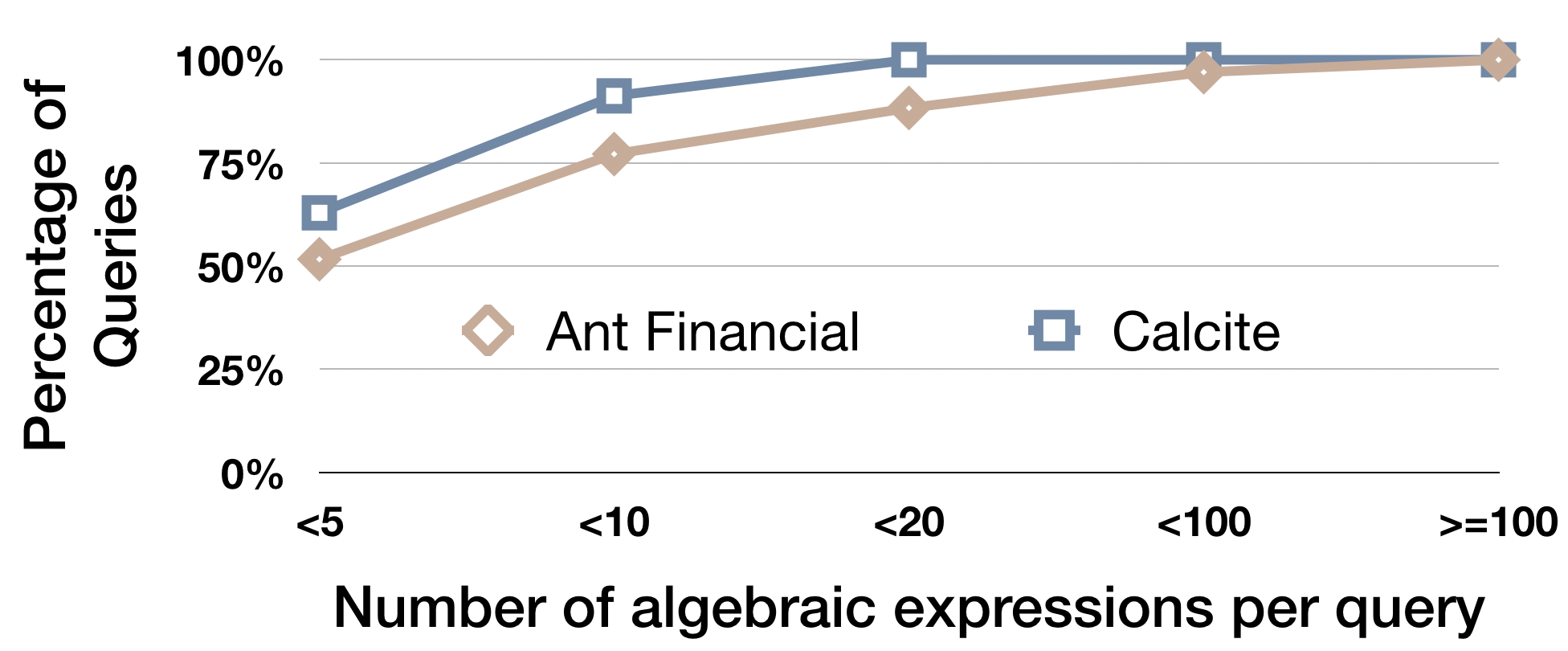}
\caption{
 \textbf{Complexity of Production Queries} -
 We quantify the complexity of production queries in the Ant Financial workload
 by measuring the number of sub-queries in each query. 
}
\label{fig:ant3}
\end{figure}

\PP{Query Complexity:}
\cref{fig:ant3} illustrates the complexity of queries in this workload.
We compute the distribution of the number of sub-queries 
in a given query (complex queries will have a larger set of sub-queries).
We found that the average number of algebraic expressions in the Ant Financial
workload (45.38) is 8$\times$ larger than that in the \calcite benchmark (5.37).

\subsection{Limitations}
\label{sec:evaluation-limitations}
In general, the problem of deciding QE is undecidable~\cite{Alfred1979}.
Among the 120 query pairs supported by \sys, it cannot prove the QE of 25
pairs.
We classify them into three categories:
(1) lack of normalization rules ($21$), and
(2) support for additional integrity constraints ($4$)
%
%We seek to address these limitations with additional engineering.

%
\PP{Normalization Rules:}
\sys verifies the cardinal equivalence of two queries only if it is able to
normalize them into the same type of queries using a set of pre-defined
semantically-equivalent rewrite rules (\autoref{sec:equivalence::construct}).
We will need to introduce additional normalization rules for queries with:
(1) union and aggregate (15/21), and
(2) join and aggregate (6/21),
Adding these rewrite rules in the normalization stage will enable \sys to 
prove the QE of these 21 pairs.
However, it will also increase the average QE verification time.
Furthermore, these rules are not required for supporting production queries 
discussed in~\autoref{sec:evaluation-summary}.

\PP{Integrity Constraints:}
\sys currently only has partial support for integrity constraints in the form
of normalization rules (\autoref{sec:approach::unf}):
(1) self join on primary key, and (2) grouping based on primary key. 
We plan to add additional rules to handle join on primary and foreign keys in
the future.
%
%We will need to encode these integrity constraints in our normalization rules.
%
%Updating these rules will enable \sys to prove the QE of 7 pairs.
%
For example, we may normalize an \ojoin operation based on a foreign key to 
an \ijoin operation.
%
%We plan to add support for integrity constraints in the future.

%% file: fig/table.tex
\renewcommand{\arraystretch}{1.2}
\begin{tabular}
{@{} l  c c  c  c | c c | c c | c c @{} }
\toprule
\shortstack{\textbf{QE} \\\textbf{Verifier}} & 
\shortstack{\textbf{Supported} \\\textbf{Semantics}} & 
\shortstack{\textbf{Supported} \\\textbf{Pairs} } &
\shortstack{\textbf{Proved} \\\textbf{Pairs} } &
\shortstack{\textbf{Average}\\\textbf{Time (s)}} &
\shortstack{\textbf{USPJ}\\\textbf{Pairs}} &
\shortstack{\textbf{Average}\\\textbf{Time (s)}} &
\shortstack{\textbf{Aggregate}\\\textbf{Pairs}} &
\shortstack{\textbf{Average}\\\textbf{Time (s)}} &
\shortstack{\textbf{Outer-Join}\\\textbf{Pairs}} &
\shortstack{\textbf{Average}\\\textbf{Time (s)}}\\
\hline
\eqs & Set & 91 & 67 & 0.15 & 28 &0.10 &32 &0.19 & 9 & 0.19 \\
\udp & Bag & 39 & 34 & N/A & 21 & N/A &11 & N/A & -- & -- \\ 
\hline
\sys (w/o norm.) & Bag & 120 & 56 & 0.02 & 31 & 0.01 & 24 & 0.04 & 3 &
0.1
\\
\sys & Bag & 120 & 95 & 0.05 & 42 & 0.05 & 44 & 0.06 & 20 & 0.08 \\
\bottomrule

\end{tabular}

%% file: fig/table2.tex
\renewcommand{\arraystretch}{1.2}
\begin{tabular}
{@{} l  | c  c  c | c  c c c c@{} }
\toprule
\shortstack{\textbf{Query}\\\textbf{Set}} & 
\shortstack{\textbf{Number of}\\\textbf{Queries}} &
\shortstack{\textbf{Queries with}\\\textbf{Overlapping Computation}} &
\shortstack{\textbf{Highest}\\\textbf{Query Frequency}} &
\shortstack{\textbf{Compared}\\\textbf{Query Pairs}} & 
\shortstack{\textbf{Equivalent}\\\textbf{Query Pairs}} & 
\shortstack{\textbf{Equivalent Pairs with}\\\textbf{Aggregate and Join}} &
\\
\hline
Set 1 & 3285 & 943 & 52 &  122900 & 3344  & 653 \\
Set 2 & 3633 & 984 & 97 &  55311 & 7225 &  4822 \\
Set 3 & 2568 & 664 & 30 & 15442 & 1521 &  356 \\
\midrule 
Total & 9486 & 2591 (27\%) &  -- & 193633 & 12090  & 5831 (48\%) \\
\bottomrule
\end{tabular}

%% file: related-work.tex
\section{Related Work}
\label{sec:related-work}
\PP{Query Equivalence:}
The state-of-the-art QE verification tools are based on either 
symbolic reasoning~\cite{zhou19} or
algebraic reasoning~\cite{chu2017cosette,chu2017sig,chu2017pldi} .
Tools based on algebraic representation include \cos and \udp that rely on 
$\mathcal{K}$-relation and $\mathcal{U}$-semi-ring theories, respectively.
Tools based on symbolic reasoning include \eqs.
We highlighted the differences between \sys and these tools
in~\autoref{sec:background}.

Prior efforts have examined the theoretical aspects of equivalence
and containment relationships between queries.
Since it is an undecidable problem~\cite{abiteboul95,ba50}, 
these efforts focused on determining categories of queries for which it is
a decidable problem:
(1) conjunctive queries~\cite{Chekuri1997}, (2) conjunctive queries with
additional constraints~\cite{diego08,horrocks00,De99}, and (3) conjunctive
queries under bag semantics~\cite{Ioannidis95}.
The problem of deciding the containment relationship between conjunctive queries 
can be reduced to a constraint satisfiability problem \cite{phokinG1998}.
Other proposals include decision procedures for proving the equivalence of a subset
of queries under set~\cite{ashok77,tannen99,Sagiv:1980} and bag
semantics~\cite{cohen99,Deutsch2008,deutsch02,popa02}.
%
%These seminal efforts have studied the theoretical aspects of proving 
%QE as opposed to practical application on real-world queries.
%
%Another line of research focuses on proving equivalence of conjunctive
%queries~\cite{tannen99,cohen02} that are generated by
%pre-defined query transformation rules.
%
%These techniques do not support complex semantically-equivalent predicates.
%
%Researchers have studied the problem of detecting overlapping
%sub-expressions~\cite{alekh18}.
%
%They take a syntactical approach.
%
%\sys instead takes a semantic approach.
%
%There are also research focuses on proving equivalence of database
%schema~\cite{albert97,albert99} and entity SQL queries~\cite{rull13}.
%
There is a large body of work on optimizing queries using materialized
views~\cite{goldstein01,rachel00,alekh18,luis2014,sanjay2000}.
These systems internally use a QE verification algorithm that relies on
syntactical equivalence and that work on a subset of \sql queries.
These systems may leverage \sys to better detect QE and deliver more optimized execution plans.

\PP{Symbolic Reasoning in DBMSs:}
Researchers have leveraged symbolic reasoning in DBMSs 
by reducing the given problem to an FOL satisfiability problem and 
then using an SMT solver to solve it.
These efforts include: 
(1) automatically generating test cases for database
applications~\cite{veanes09,veanes10,abdul08}, 
(2) verifying the correctness of database
applications~\cite{wang17,Shachar17,grossman17}, 
(3) disproving the equivalence of SQL queries~\cite{chu2017cosette}, and
%(4) proving the equivalence of SQL queries~\cite{zhou19}, 
%(4) synthesizing big data queries~\cite{schlaipfer17}, and 
(4) finding the best application-aware memory layout~\cite{yan19}.
\sys differs from these efforts in that it uses symbolic reasoning to
prove QE under bag semantics.

%% file: conclusion.tex
\section{Conclusion}
\label{sec:conclusion}
In this paper, we presented the design and implementation of \sys, an automated
tool for proving query equivalence under bag semantics using symbolic reasoning.
We reduced the problem of proving equivalence under bag semantics to that 
of proving the existence of a bijective, identity map between the output tables
of the queries for all valid input tables. 
\sys classifies queries into four categories, and leverages a set of
category-specific verification algorithms to effectively determine their
equivalence.
Our evaluation shows that \sys proves the equivalence of a larger set 
of query pairs under bag semantics compared to state-of-the-art tools based on 
symbolic and algebraic approaches.
Furthermore, it is 3$\times$ faster than \eqs that only proves equivalence under
set semantics.

%% file: appendix.tex
\section{Semantics of AR}
\label{app:ar-semantics}
We now formally define the semantics of queries, using the following
formal notation.
$\Downarrow$ is the evaluation symbol.
The left side of this symbol is an algebraic expression that is evaluated on
valid input tables $\inputs$.
The right side of this symbol is the evaluation result, which is the output
table.
All output tables are bags (\ie, can contain duplicate tuples).
A horizontal line separates the pre- and the post-conditions.
The pre-conditions on the top of the line include a set of evaluation
relations.
The post-condition on the bottom side of the line is an evaluation relation.
If all the relations in the pre-conditions hold,  
then the relation in the post-condition holds.

\begin{figure}[H]
  \begin{gather*}
  \inference[$\textsc{E-Table}$]{%
  }{%
  \evaluation{\inputs}{\TABLE{n}}{[ t | \forall t \in n ]}
  }%
  \end{gather*}
  \begin{gather*}
  \inference[$\textsc{E-SPJ}$]{%
  \vec{\expr} = {e_0,e_1,\dots,e_n}\\
  \evaluation{\inputs}{e_0}{T_0}
  \dots
  \evaluation{\inputs}{e_n}{T_n}
  }{%
    \evaluation{\inputs}{\SPJ{\vec{\expr}}{\textsc{p}}{\vec{o}}}{[(\vec{o}(t)| \forall t \in (T_0 \times \dots \times T_n),p(t)]}
  }
  \end{gather*}
  \begin{gather*}
  \inference[$\textsc{E-Agg}$]{%
   \evaluation{\inputs}{\expr}{T_0}
  }{%
  \evaluation{\inputs}{\SAggregate{\expr}{\vec{g}}{\vec{agg}}}{[\vec{agg}(t) | \forall t \in part(T_0,\vec{g})]}
  }
  \end{gather*}
  \begin{gather*}
   \inference[$\textsc{E-Union}$]{%
  \vec{\expr} = {e_0,e_1,\dots,e_n}\\
  \evaluation{\inputs}{e_0}{T_0}
  \dots
  \evaluation{\inputs}{e_n}{T_n}
  }{%
    \evaluation{\inputs}{\textsc{Union}~\vec{\expr}}{[t| \forall t \in T_0  +\dots + T_n ]}
  }
  \end{gather*}
\caption{\textbf{Semantics -- } Semantics of queries used in \sys}
\label{app::semantics}
\end{figure}

\squishitemize
\item % semantics of table queries:
Given a set of valid input tables $\inputs$, the table query returns all the tuples
in table $n$. 

\item % semantics of SPJ:
Given a set of valid input tables $\inputs$, the SPJ query first evaluates the
vector of input sub-queries on $\inputs$ to obtain a vector of input tables.
For each tuple $t$ in the cartesian product of the vector of input
tables, if $t$ satisfies the given predicate $p$, it then applies the
vector of expressions $\overrightarrow{o}$ on the selected tuple $t$
and emits the transformed tuple.

\item % semantics of aggregate:
Given a set of valid input tables $\inputs$, this aggregate query first evaluates
the input sub-query on $\inputs$ to get an input table $T_0$.
Then, it uses \texttt{part} to partition the input table $T_0$ into a set of
bags of tuples as defined by a set of group set $\vec{g}$ (tuples in
each bag take the same values for the grouping attributes).
Lastly, for each bag of tuples, it applies the vector of aggregate functions
and returns one tuple.

\item % semantics of union:
Given a set of valid input tables $\inputs$, this union query first evaluates the 
vector of input sub-queries on $\inputs$ to get a vector of input tables.
It then returns all the tuples present in the input tables, which does not eliminate duplicate tuples.

\squishend

\section{Predicate \& Proj. Expression}
\label{sec:appendix::syntax}

\sys supports the predicate and project expressions shown
in~\autoref{app::pred}.
It uses the same encoding scheme as the one employed in \eqs (described in
Section 3.4 of~\cite{zhou19}).

\begin{figure}[t]
\begin{align}
\centering
  &\exprE ::=\column{i}
         |\constant{v}
         |\nul |\bin{\exprE}{\op}{\exprE}
         |\Function{N}{\overrightarrow{\exprE}} | \CASE \label{ln:exprE} \\ 
  & \CASE :: = \Pair \exprE \\
  & \Pair :: = \WHEN{\predC}{\exprE} \Pair | \epsilon \\    
  &\op ::= + | - | \times | \div | mod \\
  &\predC ::=\binE{\exprE}{\cp}{\exprE}
         |\binL{\predC}{\logic}{\predC}
         |\Not{\predC}
         |\isNull{\exprE} \label{ln:predC} \\       
  &\cp ::=~ > | < | = | \leq | \geq\\
  &\logic ::=~ \opand |~ \opor
\end{align}
\caption{\textbf{Predicate \& Projection Expressions -- } Types of 
predicates and projection expressions supported by \sys.}
\end{figure}
\label{app::pred}

A projection expression $\exprE$ can either be a column that refer to a specific column, a constant value, $\nul$, a binary expression, 
an uninterpreted function, or an CASE expression (\autoref{ln:exprE}).
A predicate $\predC$ can either be a binary comparison between two projection expression, a binary predicate that is composed by two predicates, a not predicate and a 
predicate decide if a projection expression is $\nul$ (\autoref{ln:predC}).

% proof of soundness:
\section{Soundness of Verification}
\label{app:sound}
In this section, we give the formal proof of the soundness of \sys.
The overall proof is structured as follows:
we introduce symbolic representations of bijections over tuples in
(\autoref{app:sym-bijs}), %
prove correctness of the procedure for generating symbolic
representations (\autoref{app:sound-card}), and %
then prove correctness of the procedure for determining equivalence
(\autoref{app:sound-alg}).

\subsection{Symbolic bijections between queries}
\label{app:sym-bijs}
All definitions in this section are given with respect to arbitrary
queries $\qa$ and $\qb$, whose columns are denoted $\vec{\tabcol_1}$ and
$\vec{\tabcol_2}$, respectively.

% define cardinality preserving relations:
A \emph{cardinality-preserving} binary relation between $\qa$ and
$\qb$ is a relation
\[ R \subseteq \vec{\tabcol_1} \times \vec{\tabcol_2} \]
such that for each input table set $I$ and all tuples $(t, u) \in R$,
it holds that
\[ \tblcount{ \qa(I) }{ t } = \tblcount{ \qb(I) }{ u }
\]
where $\tblcount{ T }{ u }$ denotes the number of occurrences of tuple
$u$ in table $T$.

% introduce key correctness lemma:
Cardinality-preserving binary relations can act as witnesses of full
equivalence (see \cref{def:equivalent}) between queries.
\begin{lemma}
  \label{lemma:bij-witness}
  If the identity function is a cardinality-preserving binary relation
  between $\qa$ and $\qb$, then $\qa$ and $\qb$ are fully equivalent
  (denoted $\qa \equiv \qb$).
\end{lemma}
\begin{proof}
  Let $I$ be an arbitrary table set and let $t$ be an arbitrary tuple
  over columns $\tabcol$.
  Then 
  \[ \tblcount{ \qa(I) }{ t } = \tblcount{ \qb(I) }{ t }
  \]
  by the fact that the identity function is cardinality-preserving.
  Thus $\qa(I)$ and $\qb(I)$ are equivalent under bag semantics by the
  definition of bag semantics.
  Thus $\qa \equiv \qb$, by definition of equivalence.
\end{proof}

% define symbolic representations of binary relations:
A \emph{symbolic representation} of a binary relation
\[ R \subseteq \vec{\tabcol_1} \times \vec{\tabcol_2} \]
is an SMT formula over a vocabulary that extends $\vec{\tabcol_1}$ and
$\vec{\tabcol_2}$ such that for each $(t, u) \in R$, there is some model $m$
of $\varphi$ such that
\begin{align*}
  t = &\ m(\vec{\tabcol_1}) & u = &\ m(\vec{\tabcol_2})
\end{align*}
where $m(\vec{\tabcol_1})$ is the tuple of interpretations of each column
name in $\vec{\tabcol_1}$ (and similarly for $m(\vec{\tabcol_2})$).

% connect normalized representation to QPSR in rest of the paper
Symbolic representations of cardinality-preserving bijections can be
viewed as QPSRs (defined in \cref{sec:equivalence::check}), collapsed
into single SMT formulas.
In particular each QPSR $(C_1, C_2, c, a)$ of $\qa$ and $\qb$
corresponds to the symbolic relation
\[ \vec{\tabcol_1} = C_1 \land %
  \vec{\tabcol_2} = C_2 \land %
  c \land a
\]

Symbolic cardinality-preserving bijections can be conjoined with
equivalent constraints over column fields to form new symbolic
cardinality-preserving bijections.
In order to formalize this, we will say that for each partial
bijection $b$ between $\vec{\tabcol_1}$ and $\vec{\tabcol_2}$, each formula
$\varphi_1$ over vocabulary $\tabcol_1$, and each formula $\varphi_2$
over vocabulary $\varphi_2$, $\varphi_1$ and $\varphi_2$ are
\emph{equivalent over $b$} if $m$ is a model of $\varphi_1$ if and
only if $b(m)$ is a model of $\varphi_2$.
\begin{lemma}
  \label{lemma:sym-bij-conjoin}
  For each cardinality-preserving bijection $b$ symbolically
  represented by $\varphi$ and all $\psi_1$ over $\vec{\tabcol_1}$ and
  $\psi_2$ over $\vec{\tabcol_2}$ that are equivalent over $b$,
  \[ \varphi \land \psi_1 \land \psi_2
  \]
  is a symbolic cardinality-preserving bijection.
\end{lemma}
\begin{proof}
  $b|_{\iota(\varphi_1)}$ is is the interpretation of
  $\varphi \land \psi_1 \land \psi_2$.
  It is a cardinality-preserving bijection because it is a restriction
  of a cardinality-preserving bijection.
\end{proof}

% lemma: symbolic reprs of cardinality-preserving bijections
Because cardinality-preserving bijections can act as witnesses of
equivalence, their symbolic representations naturally can, as well.
\begin{lemma}
  \label{lemma:sym-witness}
  If there is some symbolic cardinality-preserving bijection $\varphi$
  between $\qa$ and $\qb$ that entails
  \[ \vec{\tabcol_1} = \vec{\tabcol_2} \]
  then $\qa \equiv \qb$.
\end{lemma}
\begin{proof}
  $\varphi$ represents the identity function by the assumptions that
  it represents a total function and that it logically entails a
  symbolic representation of the identity relation.
  Thus, $\qa \equiv \qb$, by \cref{lemma:bij-witness}.
\end{proof}

\subsection{Synthesizing symbolic bijections}
\label{app:sound-card}

We now prove the soundness of the procedure $\veriCard$.
The proof is defined using a set of lemmas per form of input query
(\cref{app:table-bijs}---\cref{app:union-bijs}), each of which are
predicated on assumptions that $\veriCard$ is sound on smaller
queries.
The proof for arbitrary queries combines the lemmas that concern each
form of query in a proof by induction on $\veriCard$'s input query
(\cref{app:arby-bijs}).

\subsubsection{Symbolic bijections between table queries}
\label{app:table-bijs}

We now state and prove the soundness of $\veriCardTable$, which is
given in \cref{alg:tableCheck}.
For a given pair of table queries $\TABLE{n_1}$ and $\TABLE{n_2}$, $\veriCardTable$ first checks if
two table queries have the same name.
If two table queries have the same names, then $\veriCardTable$ uses procedure $\fresh$ to create a new vector of pair of symbolic variables
based on the input table schema, and assign this new vector to $\vec{\tabcol_1}$.
$\veriCardTable$ then sets $\vec{\tabcol_1}$ is equal to $\vec{\tabcol_2}$.
$\veriCardTable$ returns the QPSR with $\vec{\tabcol_1}$ and $\vec{\tabcol_2}$, where both $\cond$ and $\assign$ are $\tru$.
If two table queries have different names, then $\veriCardTable$ returns $\nul$.

% soundness proof for table query AR's:
\begin{lemma}
  \label{lemma:table}
  If $\veriCardTable$, given two table queries $q_1 = \TABLE{n_1}$ and
  $q_2 = \TABLE{n_2}$, returns some QPSR $\varphi$, then $\varphi$ is
  symbolic cardinality-preserving bijection between $q_1$ and $q_2$.
\end{lemma}
\begin{proof}
  $\veriCardTable$ determines that $n_1 = n_2$, by the fact that
  $\veriCardTable$ only returns a QPSR if $n_1 = n_2$ and by the
  assumption that $\veriCardTable$ returns a QPSR.
  The QPSR returned by $\veriCardTable$ is the symbolic representation
  of the identity relation, and is thus a symbolic
  cardinality-preserving bijection.
\end{proof}

\subsubsection{Symbolic bijections between SPJ queries}
\label{app:spj-bijs}
We now formalize and prove the correctness of $\veriCardSPJ$ (see
\cref{sec:equivalence::construct:::spj}).

\begin{lemma}
  \label{lemma:spj}
  For vectors of queries $\overrightarrow{e_1}$ and
  $\overrightarrow{e_2}$, %
  if $\veriCard(e_1', e_2')$ is an QPSR $\varphi'$ only if $\varphi'$
  is a symbolic cardinality-preserving bijection between $e_1'$ and
  $e_2'$ for all $e_1' \in \overrightarrow{e_1}$ and
  $e_2' \in \overrightarrow{e_2}$, then %
  then $\veriCardSPJ$, given queries
  $\qa = \SPJ{\overrightarrow{e_1}}{p_1}{\overrightarrow{o_1}}$ and
  $\qb = \SPJ{\overrightarrow{e_2}}{p_2}{\overrightarrow{o_2}}$,
  returns a QPSR $\varphi$ only if $\varphi$ is a symbolic
  cardinality-preserving bijection between $\qa$ and $\qb$.
\end{lemma}
\begin{proof}
  % prove that QPSR is a symbolic bijection
  There is some $\mathsf{QPSR}$ in
  $\veriCardExpr(\overrightarrow{e_1}, \overrightarrow{e_2})$, by the
  definition of $\veriCardSPJ$ (\cref{alg:SPJCheck},
  \cref{ln:spj-joins}, \cref{ln:spj-iter}, \cref{ln:spj-ret-qpsr}, and
  \cref{ln:spj-ret-null}) and the semantics of $\veriCardExpr$.
  $\mathsf{QPSR}$ represents a symbolic cardinality-preserving
  bijection $\varphi$ between the Cartesian product of
  $\overrightarrow{e_1}$ and the Cartesian product of
  $\overrightarrow{e_2}$, by the assumption that
  $\veriCard(e_1', e_2')$ is a symbolic cardinality-preserving
  bijection for all $e_1 \in \overrightarrow{e_1}$ and
  $e_2 \in \overrightarrow{e_2}$.

  % prove that conjunction and projection is a sym bij:
  Thus $\varphi \land \cond_1 \land \cond_2$ is a symbolic
  cardinality-preserving bijection by \cref{lemma:sym-bij-conjoin}
  applied to $\cond_1$ and $\cond_2$ and the definition of
  $\veriCardSPJ$ (\cref{ln:prove} and \cref{ln:spj-ret-qpsr}).
  Furthermore, it is a symbolic cardinality-preserving bijection of
  between the selection of the Cartesian products of
  $\overrightarrow{e_1}$ on $p_1$ and the Cartesian product of
  $\overrightarrow{e_2}$ on $p_2$.
  Thus the QPSR returned by $\veriCardSPJ$ is a symbolic
  cardinality-preserving bijection between $\qa$ and $\qb$ by the
  definition of $\veriCardSPJ$
  (\cref{ln:spj-def-assn}---\cref{ln:spj-ret-qpsr}).
\end{proof}

\subsubsection{Symbolic bijections between aggregate queries}
\label{app:agg-bijs}
We now formalize and prove the correctness of $\veriCardAgg$ (see
\cref{sec:equivalence::construct:::agg}).
\begin{lemma}
  \label{lemma:agg}
  If $\veriCardAgg$, given aggregate queries of the form
  $q_1 =
  \SAggregate{e_1}{\overrightarrow{g_1}}{\overrightarrow{a_1}}$ and
  $q_2 =
  \SAggregate{e_2}{\overrightarrow{g_2}}{\overrightarrow{a_2}}$
  returns a QPSR $\varphi$, and if $\veriCard$ given $e_1$ and $e_2$,
  only returns a QPSR if it is a symbolic cardinality-preserving
  bijection, then $\varphi$ is a symbolic cardinality-preserving
  bijection.
\end{lemma}
\begin{proof}
  In this proof, let $\tabcol_1$ and $\tabcol_2$ denote the columns of
  $e_1$ and $e_2$.
  $\veriCard$, given $e_1$ and $e_2$, returns a symbolic
  cardinality-preserving bijection $b$, by the definition of
  $\veriCardAgg$ (see \cref{alg:AggCheck}, \cref{ln:agg-chk}) and the
  assumption that if $\veriCard$ returns a QPSR, then it is a symbolic
  cardinality-preserving bijection.
  The QPSR returned by $\veriCardAgg$ is a symbolic
  cardinality-preserving bijection whose interpretation is the
  composition of bijections $b = b_0^{-1} \compose b_E \compose b_1$
  (where $b_0^{-1}$ denotes the inverse of bijection $b_0$), which are
  defined as follows.
  Let $b_0$ be the bijection from each equivalence class of
  $\tabcol_1$ induced $\overrightarrow{g_0}$ to its image under the
  aggregate functions $\overrightarrow{a_0}$, and similarly for $b_1$.

  % define bijection between equivalence classes:
  Let $b_E$ be the relation from each equivalence class $E_0$ of
  $\tabcol_1$ induced by $\overrightarrow{g_1}$ to equivalence class
  $E_1$ of $\tabcol_2$ induced by $\overrightarrow{g_2}$ if there is
  some $t \in E_0$ such that $b(t) \in E_1$.
  $b_E$ is in fact a bijection, by the definition of $\veriCardAgg$
  (see \cref{ln:chk-gps} and \cref{ln:agg-ret-qpsr}, which checks that
  the equivalence classes induced by $\overrightarrow{g_1}$ and
  $\overrightarrow{g_2}$ preserve $b$).

  % connect to code of veriCardAgg:
  $b$ is the interpretation of the QPSR returned by $\veriCardAgg$ by
  the semantics of $\freshAgg$ and $\constructAggCall$ and the
  definition of $\veriCardAgg$ (see
  \cref{ln:aggNewSymbolicTuples1}---\cref{ln:agg-ret-qpsr}).
\end{proof}

\subsubsection{Symbolic bijections between union queries}
\label{app:union-bijs}

We now formalize and prove soundness of $\veriCardUnion$ (see
\cref{sec:equivalence::construct:::union}).

% soundness proof for Union ARs
\begin{lemma}
  \label{lemma:union}
  If $\veriCardUnion$, given $q_1 = \SUnion{\overrightarrow{e_1}}$ and
  $q_2 = \SUnion{\overrightarrow{e_2}}$ returns some QPSR $\varphi$
  and if $\veriCard(e_1', e_2')$ returns a QPSR then it is a symbolic
  cardinality-preserving bijection for all $e_1' \in e_1$ and
  $e_2' \in e_2$, then $\varphi$ is a symbolic cardinality-preserving
  bijection between $q_1$ and $q_2$.
\end{lemma}
\begin{proof}
  The result of
  $\veriCardExpr(\overrightarrow{e_1}, \overrightarrow{e_2})$ is a
  vector of QPSRs, by the definition of $\veriCardUnion$ and the
  assumption that $\veriCardUnion$ returns some QPSR.
  Each of the QPSRs is a symbolic cardinality-preserving bijection, by
  the inductive hypothesis.
  Thus, the QPSR $\varphi$ returned by $\veriCardExpr$ is a symbolic
  cardinality-preserving bijection, by the definition of $\PN$.
\end{proof}

\subsubsection{Symbolic bijections between arbitrary queries}
\label{app:arby-bijs}

% soundness of verifying cardinality:
The soundness of $\veriCard$ follows from the correctness properties
satisfied for each type of query, given above.

% lemma: soundness of veriCard:
\begin{lemma}
  \label{lemma:vericard-sound}
  If $\veriCard$, given queries $\qa$ and $\qb$, returns some QPSR
  $\varphi$, then $\varphi$ is a symbolic representation of a
  cardinality-preserving bijection between $\qa$ and $\qb$.
\end{lemma}
\begin{proof}
  Proof is by induction on $\qa$ (the choice of $\qa$ versus $\qb$ is
  basically arbitrary).
  Each case on the form of $\qa$ is proved by applying one the above
  lemmas, to the inductive hypothesis in cases when $\qa$ is a
  composite of sub-queries.
  In particular:
  % cases of the form of QA:
  \begin{itemize}
  \item % table query:
    If $\qa$ is a \emph{table name}, then the proof follows
    immediately from \cref{lemma:table}.
  \item % SPJ:
    If $\qa$ is an \emph{SPJ} query, the proof follows immediately from
    \cref{lemma:spj}.
  \item % aggregate:
    If $\qa$ is a \emph{aggregate} query of the form
    $\SAggregate{\qa'}{\overrightarrow{g_1}}{\overrightarrow{a_1}}$,
    then $\qb$ is a union query of the form
    $\SAggregate{\qb'}{\overrightarrow{g_2}}{\overrightarrow{a_2}}$,
    by the assumption that $\veriCard$ returns some QPSR.
    The proof follows from applying \cref{lemma:agg} to the inductive
    hypothesis on queries $\qa'$ and $\qb'$.
  \item % union:
    If $\qa$ is a \emph{union} query of the form
    $\SUnion{\overrightarrow{\qa'}}$, then $\qb$ is a union query of
    the form $\SUnion{\overrightarrow{\qb'}}$, by the assumption that
    $\veriCard$ returns some QPSR.
    The proof follows from applying \cref{lemma:union} to the
    inductive hypothesis on all sub-queries in $\qa'$ and $\qb'$.
  \end{itemize}
\end{proof}

\subsection{Soundness of \sys}
\label{app:sound-alg}

% lemma: full equivalence reduces to validity of a symbolic
% representation
With the soundness of $\veriCard$ established, we are prepared to
state and prove the soundness of $\sys$.

\begin{algorithm}[t]
  \SetKwInOut{Input}{Input}
  \SetKwInOut{Output}{Output}
  \SetKwProg{myproc}{Procedure}{}{}
  \Input{A pair of queries \qa and \qb}
  \Output{An decision if two queries are fully equivalent}
  \myproc{$\sys(\qa,\qb)$}{
       $\qa' \gets \mathsf{normalize}(\qa)$ \;
       $\qb' \gets \mathsf{normalize}(\qb)$ \;
       $\varphi \gets \veriCard(\qa', \qb')$ \;
       \uIf{$\varphi \not= \nul$}{
           \Return{ $\mathsf{isValid}(\varphi \implies \vec{\tabcol_1} = \vec{\tabcol_2})$ }
       }
       \lElse { \Return{ \fal } }
  } \caption{$\sys$: an equivalence verifier.}
\end{algorithm}

Given a pair of queries \qa and \qb, \sys uses the procedure $\mathsf{normalize}$ to converts each queries to algebraic expressions, and uses a set of semantic preserving rewrite
rules to normalize them.
These semantic preserving rules are defined in \cref{sec:approach}.
Then \sys uses the procedure $\veriCard$ to constructs the QPSR of two normalized queries \qa' and \qb'.
If $\veriCard$ returns the QPSR (i.e., the QPSR is not \nul), then \sys returns if the formula is valid.
If $\veriCard$ doesn't return the QPSR (i.e., the QPSR is \nul), then \sys returns $\fal$.

% main soundness lemma:
The soundness of \sys is formalized and proved in the following
theorem:
\begin{theorem}
  \label{theorem:sound}
  If \sys, given queries $\qa$ and $\qb$, returns $\mathsf{true}$,
  then $\qa \equiv \qb$.
\end{theorem}
\begin{proof}
  $\qa'$ is equivalent to $\qa$ and $\qb'$ is equivalent to $\qb$
  because the normalization rules that it applies preserve semantics (\cref{sec:approach}).
  QPSR is a symbolic cardinality-preserving bijection by the
  assumption that $\sys$ returns $\mathsf{true}$ and thus QPSR is not
  $\mathsf{null}$, and by \cref{lemma:vericard-sound}.
  The formula
  \[ \varphi \implies %
     \vec{\tabcol_1} = \vec{\tabcol_2} \]
  is valid by the assumption that $\sys$ returns $\true$.
  Thus $\varphi$ entails
  \[ \vec{\tabcol_1} = \vec{\tabcol_2} \]
  by the semantics of SMT.
  Thus, $\qa \equiv \qb$ by \cref{lemma:sym-witness}.
\end{proof}